\documentclass[11pt]{article}
\usepackage{amsmath}
\usepackage{fancyhdr}
\usepackage{amssymb}
\usepackage{amsthm}
\usepackage{graphicx}
\usepackage{varioref}
\usepackage{verbatim} 
\usepackage{multicol}
\usepackage{lmodern}
\usepackage{enumerate}
\usepackage[normalem]{ulem}
\usepackage{caption}
\usepackage{subfig}
\usepackage[T1]{fontenc}
\usepackage[margin=1in]{geometry}
\usepackage{fancyhdr}
\usepackage{authblk}

\usepackage{mathrsfs}

\usepackage{url}
\usepackage{xspace}
\usepackage{thm-restate}

\usepackage{hyperref}
\hypersetup{
    bookmarksnumbered=true, 
    unicode=false, 
    pdfstartview={FitH}, 
    pdftitle={}, 
    pdfauthor={Shelby Kimmel}, 
    pdfsubject={}, 
    pdfcreator={}, 
    pdfproducer={}, 
    pdfkeywords={}, 
    pdfnewwindow=true, 
    colorlinks=true, 
    linkcolor=blue, 
    citecolor=blue, 
    filecolor=blue, 
    urlcolor=blue 
}

\newtheorem{theorem}{Theorem}
\newtheorem{definition}{Definition}
\newtheorem{lemma}{Lemma}

\newtheorem{corollary}[theorem]{Corollary}

\theoremstyle{definition}
\newtheorem{alg}{Algorithm}

\theoremstyle{remark}

\newcommand{\ket}[1]{|#1\rangle}

\newcommand{\tth}[0]{\textsuperscript{th}}

\usepackage{xcolor}

\DeclareMathAlphabet{\matheu}{U}{eus}{m}{n}

\newcommand{\sop}[1]{{\mathcal #1}}

\newcommand{\I}{{\mathbb I}}

\newcommand{\braket}[2]{\langle{#1}|{#2}\rangle}
\newcommand{\ketbra}[2]{|{#1}\rangle\!\langle{#2}|}

\newcommand{\no}{\nonumber\\}

\begin{document}

\title{{\huge{\bf{Oracles with Costs}}}}

\author[1,2]{{\Large{Shelby Kimmel}}}
\author[2]{{\Large{Cedric Yen-Yu Lin}}}
\author[2]{{\Large{Han-Hsuan Lin}}}
\affil[1]{Joint Center for Quantum Information and Computer Science,
University of Maryland}
\affil[2]{Center for Theoretical Physics, Massachusetts Institute of Technology}
 
\date{}

\maketitle
\begin{abstract}
While powerful tools have been developed to analyze quantum query
complexity, there are still many natural problems that do not fit
neatly into the black box model of oracles. We create a new model that
allows multiple oracles with differing costs. This model captures more
of the difficulty of certain natural problems. We test this model on a
simple problem, Search with Two Oracles, for which we create a quantum
algorithm that we prove is asymptotically optimal. We further give
some evidence, using a geometric picture of Grover's algorithm, that
our algorithm is exactly optimal.
\end{abstract}

\thispagestyle{fancy}
\rhead{MIT-CTP/4637}
\renewcommand{\headrulewidth}{0pt}
\renewcommand{\footrulewidth}{0pt}


\section{Introduction}

The standard oracle model is a powerful paradigm
for understanding quantum computers. Tools such
as the adversary semidefinite program \cite{HLS07,LMRSS11}, learning graphs
 \cite{B12,BR14}, 
and the polynomial method \cite{BBC+01} allow us to accurately characterize the 
quantum query complexity \cite{A00,BBBV97} of many problems of interest.

However, the oracle model does not capture the full power
or challenges of quantum computing. For example,
 problems such as $k$-SAT do not fit easily
into the oracle model. Additionally, 
while the query complexity of the hidden
subgroup problem is known to be polynomial in the size of the problem \cite{EHK04}, for some
non-abelian groups there is no efficient algorithm.

In this paper, we describe a variation of the oracle model. We
have access to two oracles, rather than a single oracle\footnote{The model can easily be extended
to more than two oracles, but for simplicity, we limit ourselves to two.}, but 
one oracle is more expensive to use. In the standard oracle model, the figure
of merit is the query complexity, which is the minimum number of queries needed to an oracle
to evaluate a function. In our model, the figure of merit is the {\it{cost complexity}},
which is the minimum cost needed to evaluate a function using multiple oracles with different
costs.

To motivate this model, we consider the following fact: in some search problems we want to find an
element in a set that satisfies a property that is expensive to test. However, often another less expensive test is
available that can narrow down the search range but is not conclusive. We give three examples of problems where such
less expensive, less conclusive tests are natural. In each
example, Test 1 is more expensive to run but is conclusive, while Test 2 is cheaper to run but allows some non-solutions to pass.

\begin{itemize}
\item  In the problem of $k$-SAT on $n$ bits, we would like to find an assignment $x\in\{0,1\}^n$ such that all clauses are satisfied. 
Consider an algorithm for $k$-SAT that runs two types of tests
on a possible assignment $x$:
\begin{enumerate}
\item Check whether all clauses are satisfied.
\item Check whether some subset of the clauses are satisfied.
\end{enumerate} 
\item  Given a graph $A$ and a set of graphs $\{B_1,\cdots,B_p\}$, we would like to find a graph $B_i$ isomorphic to $A$. Consider 
an algorithm that runs two types of tests on a graph $B_i$:
\begin{enumerate}
\item Check whether $B_i$ is isomorphic to $A$ (say by brute force search).
\item Check whether the adjacency matrices of $B_i$ and $A$ have the same spectrum.
\end{enumerate} 
\item 
In the decision variant of the traveling salesman problem, given a
positively weighted $N$-graph $G$ and a positive number $b$, we would
like to find a tour of the vertices of $G$ that uses cost no more than
$b$. Given a partial tour of length $N/2$, we can run two types of
tests:
\begin{enumerate}
\item Check whether the partial tour can be completed to an $N$-vertex tour that has cost at most $b$, by using brute force search.
\item Check whether the sum of the weights of the $N/2$ edges traversed in the partial tour is bigger than $b$.
\end{enumerate}
\end{itemize}

In all three examples, 
the two tests can be implemented as unitaries $\sop O_1, \sop O_2$
that act as
$\sop O_i\ket{x}\ket{y}=\ket{x}\ket{y\oplus f_i(x)}.$
Here $f_i(x)=1$ if assignment $x$ passes Test $i$ and $f_i(x)=0$
otherwise. These two unitaries will play the role of oracles with different
costs.

%

None of the problems listed above are typically thought of as oracle problems, because 
in each problem, there is more information than can easily be incorporated
into a single oracle. However, with multiple oracles, the information
 can be distributed among different oracles. Using
different costs for different oracles allows us to include
information about the time required to access information. We see that
cost complexity can capture certain aspects of a problem that can
 not be easily accounted for in the standard oracle model;
 we hope this model will provide new insight into problems previously thought
beyond the tools of query algorithms. We note that we do not expect
these techniques to allow us to solve NP-complete problems in polynomial time. 
Rather, our goal is to potentially improve upon existing exponential time algorithms,
and create connections between standard oracle problems and problems that seem far
from typical oracle problems.

Problems such as those described above can easily be recast into an 
oracle problem, which we call Search with Two Oracles (STO).
In this work, we focus on the problem of STO. We tightly characterize the quantum
cost complexity of this problem, and give several techniques for putting lower
bounds on quantum cost complexity.  We also show that the cost
complexity of STO is the same whether or not the oracles can be
accessed using a control operation; that is, accessing the oracles in
superposition gives no added power.

We also attempt to exactly bound (rather than asymptotically bound)
the cost complexity of STO. Usually, one is not particularly
interested in proving exact optimality, but we have several reasons
for wanting to explore this problem.  Few quantum algorithms are known to be exactly
optimal; Grover's
algorithm and parity are two examples~\cite{DH09,BBC+01}. STO is a very simple extension of a standard
search problem, so it seems like a good candidate problem for
obtaining another exact lower bound. Proving that our algorithm is exactly optimal would provide evidence
that amplitude amplification is exactly optimal in the case of no additional 
structure (i.e. when we treat the base algorithm as a black box). Additionally, while we can obtain
asymptotically tight bounds for the problem of STO, for a simple
extension of STO to $\log N$ oracles (where $N$ is the size of the 
search space), these techniques fail.  However,
if we could obtain tighter bounds for STO, we should be able to get a
better characterization of the cost complexity for these more
complex problems.

Finally, we compare the quantum cost complexity of STO to the classical
cost complexity. We show a polynomial reduction in cost for the
quantum version. Moreover, we show that the optimal quantum and
classical algorithms behave qualitatively differently, highlighting
the power of quantum algorithms.

In Section \ref{sec:oraclemodel}, we describe cost complexity and
define STO. In Section \ref{sec:algo}, we describe
optimal quantum algorithms for STO, and in Section \ref{sec:lowerbound}, 
we put lower bounds on the cost complexity of STO. Finally,
we look at the classical cost complexity
of STO in Section \ref{sec:classical}.


\section{Cost Complexity, STO, and Relation to Previous Work}\label{sec:oraclemodel}

Cost complexity is very closely related to query complexity. For background
on query complexity, see \cite{A00,BBBV97}.

We first define cost complexity. In the following, we use the notation
$[N]\equiv\{1,\dots,N\}.$ 
Given the input  $(f_1,f_2)\in D$, which is a pair of functions
$f_1,f_2:[N]\rightarrow\{0,1\}$, we want to calculate $F$ where
$F:D\rightarrow\{0,1\}$.
 Let $f_1$ be associated with cost $c_1$
and $f_2$ be associated with cost $c_2.$ Depending on the type of algorithm
(e.g. classical, quantum), these two functions
are accessed in different ways. 

In the classical setting, consider a  randomized classical algorithm
$\sop A_c$ for $F$ that makes $q_1$ queries to $f_1$, and $q_2$
queries to $f_2$. Then the cost of this algorithm is
\begin{align}
\textrm{Cost}(\sop A_c)=q_1c_1+q_2c_2.
\end{align}
Let $\mathscr{A}_{c,\epsilon}$ be the set of randomized classical
algorithms that solve $F$ with success probability at least
$1-\epsilon$ on all inputs in $D.$ Then the {\it{classical randomized cost complexity
(RCC)}} of $F$ is
\begin{align}
RCC_\epsilon(F)=\min_{\sop A_{c}\in\mathscr{A}_{c,\epsilon}}\textrm{Cost}(\sop A_c).
\end{align}

In the quantum setting, let $\sop O_1$ and $\sop O_2$ be unitaries
acting on the Hilbert space $\mathbb{C}^N$ with standard basis states
$\ket{i}$ for $i\in[N]$ as $\sop O_j\ket{i}=(-1)^{f_j(i)}\ket{i}$ for
$j\in{1,2}.$ Consider a quantum algorithm $\sop A_q$ that at each
time step, can apply $\sop O_1$ or $\sop O_2$ or some other unitary that is
independent of $f_1$ and $f_2$, and which
makes $q_1$
queries to $\sop O_1$ and $q_2$ queries to $\sop O_2.$
Then the cost of the algorithm $\sop A_q$ is
\begin{align}
\textrm{Cost}(\sop A_q)=q_1c_1+q_2c_2.
\end{align}
Let $\mathscr{A}_{q,\epsilon}$ be the set of quantum algorithms
that solve $F$ with success probability at least $1-\epsilon$ on all inputs in $D$. Then the 
{\it{quantum cost complexity (QCC)}} of $F$ is
\begin{align}
QCC_\epsilon(F)=\min_{\sop A_{q}\in\mathscr{A}_{q,\epsilon}}\textrm{Cost}(\sop A_q).
\end{align}

Finally, we consider quantum algorithms that can access oracles
in superposition. Let $\sop O_1$ and $\sop O_2$ be as above, and let
$\sop O_0=\I,$  the $N\times N$ identity matrix. We now consider a
quantum algorithm that has access to a controlled operation $C\sop O$ that acts
on the  the Hilbert space   $\mathbb C^3\otimes\mathbb C^{N}\otimes \mathbb{C}^V$ ($\mathbb{C}^V$ 
is a workspace register) with
standard basis states $\ket{b}\ket{i}\ket{v}$ for $i\in[N]$, $v\in[V]$, and $b\in\{0,1,2\}$ as
$C\sop O\ket{b,i}=\ket{b}\sop O_b\ket{i}\ket{v}.$ Suppose the encoded functions are $f_1$ and $f_2$. Then if an algorithm
$\sop A_{qs}$ applies $C\sop O$ a total of $T$ times over the course
of the algorithm to states
\begin{align}
\ket{\eta^t_{f_1,f_2}}=\sum_{b=0}^2 \sum_{i=1}^N \sum_{v=1}^V \alpha_{f_1,f_2}^t(b,i,v)\ket{b,i,v}
\end{align}
for $t\in[T]$, the cost of the algorithm is 
\begin{align}
\textrm{Cost}(\sop A_{qs})=\max_{f_1,f_2}\sum_{t=1}^T&\kappa(\eta^t_{f_1,f_2}) \text{ where }\nonumber\\
 \kappa(\eta^t_{f_1,f_2})=
&\begin{cases}
c_1 \text{ if } \sum_{i,v} |\alpha_{f_1,f_2}^t(1,i,v)|^2 \neq 0,\\
c_2 \text{ if } \sum_{i,v} |\alpha_{f_1,f_2}^t(1,i,v)|^2 = 0 \text{ and } \sum_{i,v} |\alpha_{f_1,f_2}^t(2,i,v)|^2 \neq 0,\\
0 \text{ if } \sum_{i,v} |\alpha_{f_1,f_2}^t(1,i,v)|^2 = 0 \text{ and } \sum_{i,v} |\alpha_{f_1,f_2}^t(2,i,v)|^2 = 0.
\end{cases}
\end{align} 
Let $\mathscr{A}_{qs,\epsilon}$ be the set of quantum algorithms using
$C\sop O$ that solve $F$ with success probability at least
$1-\epsilon$ on all inputs in $D.$  Then the {\it{controlled quantum cost complexity (ConQCC)}} of $F$ is
\begin{align}
ConQCC_\epsilon(F)=\min_{\sop A_{qs}\in\mathscr{A}_{qs,\epsilon}}\textrm{Cost}(\sop A_{qs}).
\end{align}
The controlled quantum cost complexity is closely related to the time required in the model of variable
times introduced by Ambainis in \cite{A10}.

Note that 
\begin{align}
ConQCC_\epsilon(F)\leq QCC_\epsilon(F)\leq RCC_\epsilon(F).
\label{eq:compare_costs_def}
\end{align}

For any of the cost complexities described above, if we do not include a subscript
$\epsilon,$  then the cost is assumed to apply for the case $\epsilon=1/3$.

Now that we have defined cost complexity, we introduce
 the problem of STO as a testbed for tools and ideas that  can
hopefully be applied to more complex problems. More formally, we give
the definition of STO: 
\begin{definition}[Search with Two Oracles
(STO)]\label{defi:sto} Let $N$ and $M$ be known positive integers and
let $S\subseteq [N]$ be an unknown set. There might or
might not exist a special item $i_*$. If $i_*$ exists, then one is
promised that $i_*\in S$ and $|S|=M$. If $i_*$ doesn't exist, the size
of S is arbitrary. Let $f_*$ and $f_S$ be two functions with domain $[N]$
and range $\{0,1\}$ such that 
\begin{align}
&
f_*(i)=
\begin{cases}
1 &\text{ if }i=i_*\\
0&\text{ if }i\neq i_* \text{ or }i_*\text{ doesn't exist}.
\end{cases}
&
f_S(i)=
\begin{cases}
1 &\text{ if }i\in S\\
0&\text{ if }i\notin S.
\end{cases}
\end{align}
Then STO$(f_*,f_S)=1$ if $i_*$ exists, and $0$ otherwise. 
$c_*$ is the cost associated with $f_*$ and $c_S$ is the cost associated with $f_S,$ with $c_*\geq c_S.$
\end{definition}
\noindent $c_S$ and $c_*$ are assumed to depend on $N$ and $M,$ but our results hold for any form of that dependence,
so we leave off any explicit relationship.

Cost complexity, and STO in particular, are related to several existing oracle problems. In the
problem of STO, the function $f_S$ can be thought of as providing
extra information or advice about the function $f_*$.
There have been several studies in which access to a single oracle is
supplemented with some extra information that can come in the form of
another oracle or classical information, e.g. \cite{M11,NABT14}.
Previous works \cite{ARU14,M11} have considered multiple oracles, but
not with costs. Furthermore, the additional advice oracles considered
in these works tend to be somewhat unnatural, and are tailored to the
specific problems considered. As mentioned, $ConQCC$ is related to
the model of variable costs studied by Ambainis, in which he
considered a single oracle that has different costs for querying
different items \cite{A10}. We also note that Cerf et al. \cite{CGW00} consider
similar quantum algorithms in the context of constraint satisfaction
problems, but they do not approach the problem from an oracular perspective.

\section{Quantum Algorithms for STO}\label{sec:algo}
We now describe quantum algorithms for solving STO\footnote{For the purpose of
describing these algorithms, we assume that $i_*$ exists. A single application
of $O_*$ at the end of the algorithm can be used to check (with appropriate
probability) whether or not $i_*$ exists, at a cost of $c_*$.}. These
algorithms use the oracles $\sop O_*$ and $\sop O_S$ directly,
rather than the controlled version (i.e. $C\sop O$) of these oracles. All of our algorithms
can be viewed as examples of amplitude amplification. Recall
\begin{theorem}[Amplitude Amplification \cite{BHMT00}]\label{thm:aa}
Let $T\subset[N]$, $\alpha\in[0,1]$, and let $\sop O^T$ be an quantum oracle that marks the
elements of $T$. We define 
\begin{align}
\ket{T}=\frac{1}{\sqrt{|T|}}\sum_{i\in T}\ket{i}.
\end{align}
Given an algorithm $\sop A$ that acts on a state $\ket{\psi_0}$
and produces a state $\ket{\psi_{\sop A}}$ such that
$\left|\braket{T}{\psi_{\sop A}}\right|=p$, one can create a new
algorithm $\sop B$ that
applies $\sop O^T$, $\sop A$, and $\sop A^{-1}$ each
\begin{align}
\tau= \left\lceil\frac{\arcsin\sqrt{1-\alpha} - \arcsin p} {2\arcsin p}\right\rceil
\end{align}
times, and which acts on the initial state $\ket{\psi_0}$ and produces a state $\ket{\psi_{\sop B}}$ such that
\begin{align}
\ket{\psi_{\sop B}}=\sqrt{1-\alpha}\ket{T}+\sqrt{\alpha}\ket{T^\perp},
\end{align}
where $\braket{T}{T^\perp}=0$ and 
$\ket{T^\perp}\in \textrm{Span}\left(\ket{T},\ket{\psi_A}\right)$.
\end{theorem}
This gives us the following Corollary:
\begin{corollary}\label{corr:aacost}
Let $\sop A$ and $\tau$ be as in Theorem \ref{thm:aa}, and assume $\sop O^T$ has
cost $c_T$ while $\sop A$ and $\sop A^{-1}$ have cost $c_{\sop A}$.
Then there exists a algorithm $\sop B$ that applies $\sop O^T$, $\sop A$ and $\sop A^{-1}$
not in superposition, and produces the state $\ket{T}$ with
probability $1-\epsilon$ such that
\begin{align}
\emph{Cost}(\sop B)=\tau\left(c_T+2c_{\sop A}\right).
\end{align}
\end{corollary}

In the following, we describe three algorithms for STO. 
 We consider the limit that $M,N/M\rightarrow \infty$ to simplify our analysis, but 
this limit still captures the essential behavior of the algorithms. We use the following notation:
\begin{align}
\ket{N}=& \frac{1}{\sqrt N} \sum_{i=1}^N\ket{i},\no
\ket{S} =& \frac{1}{\sqrt{M}} \sum_{i\in S} \ket{i}.
\end{align}
We have a slight abuse of notation, since $\ket{N}$ could refer
either to the equal superposition state, or the $N\tth$ standard basis
state. However, whenever we write $\ket{N}$, we will always mean the
equal superposition state.

The first algorithm we consider ignores $\sop O_S$ and performs a Grover search for $i_*$ using $\sop O_*$: 

\begin{alg} (Grover's Search) \label{alg:direct} \hfill

Prepare the state 
$\ket{N}$
at cost 0. Set $\sop A$ equal to the identity. 
Then by Corollary \ref{corr:aacost}
there exists an algorithm $\sop B$ that produces the state $\ket{i_*}$ with probability $1-\epsilon$ with cost
\begin{align}
c_*\left\lceil\frac{\arcsin\sqrt{1-\epsilon} - \arcsin \frac{1}{\sqrt{N}}} {2\arcsin \frac{1}{\sqrt{N}}}\right\rceil.
\end{align}
In the limit of $N\rightarrow\infty,$ the cost becomes
\begin{align}
c_*\arcsin\sqrt{1-\epsilon}\sqrt{N}.
\end{align}
\end{alg}

However, if $\sop O_S$ comes to us cheaply, we would like to take advantage of it:
The following
algorithm first rotates $\ket{N}$ to $\ket{S}$ (using $\sop O_S$),
and then rotates $\ket{S}$ to $\ket{i_*}$  (using both $\sop O_S$ and
$\sop O_*$).

\begin{alg} \label{alg:entire} \hfill

Prepare the state $\ket{N}$ at cost 0. Set $\sop A$ equal to the identity. 
Since $|\braket{N}{S}|=\sqrt{M/N}$, by Corollary \ref{corr:aacost}
there exists an algorithm $\sop B$ that with probability $1$ produces the state $\ket{S}$ at cost
\begin{align}
c_S\left\lceil\frac{\left(\frac{\pi}{2} - \arcsin \sqrt{\frac{M}{N}}\right)}{2 \arcsin \sqrt{\frac{M}{N}}}\right\rceil.
\end{align}
Now $|\braket{i_*}{S}|=\sqrt{1/M}$, so using Corollary \ref{corr:aacost} again, there exists an algorithm $\sop C$
that with probability $1-\epsilon$ produces the state $\ket{i_*}$ at cost
\begin{align}
\left\lceil\frac{\arcsin\sqrt{1-\epsilon} - \arcsin \frac{1}{\sqrt{M}}} {2\arcsin \frac{1}{\sqrt{M}}}\right\rceil
\left(c_*+2c_S\left\lceil\frac{\left(\frac{\pi}{2} - \arcsin \sqrt{\frac{M}{N}}\right)}{2 \arcsin \sqrt{\frac{M}{N}}}\right\rceil\right).
\end{align}
 Dropping terms of size at most $O(M^{-1/2})$ or $O\left((M/N)^{1/2}\right)$ of  the zeroth 
order terms, the cost becomes
\begin{align}
\frac{\arcsin\sqrt{1-\epsilon}}{4}\left(2c_*\sqrt{M}+\pi c_S\sqrt{N}\right).
\end{align}
\end{alg}

Combining Algorithms  \ref{alg:direct} and \ref{alg:entire}, we have that
\begin{align}\label{eq:cost12}
QCC(STO)&=O\left(\min\left\{c_*\sqrt{N},c_*\sqrt{M}+c_S\sqrt{N}\right\}\right)\nonumber\\
&=O\left(\max\left\{c_*\sqrt{M},c_S\sqrt{N}\right\}\right).
\end{align}

 In Section \ref{sec:lowerbound}, we will show that this cost (Eq.
(\ref{eq:cost12})) is asymptotically optimal. This means that Algorithm
\ref{alg:entire} is always asymptotically optimal, although Algorithm \ref{alg:direct} has lower cost when $c_* \approx c_S$.
However, it turns out that  there is an algorithm that has lower cost
than either Algorithm \ref{alg:direct} or \ref{alg:entire}.   In
Section \ref{sec:lowerbound}, we give evidence that this final
algorithm, which we call the Hybrid Algorithm,   is not just
asymptotically optimal, but exactly optimal.

The two algorithms we have so far presented can be summarized as
follows: Algorithm \ref{alg:direct} directly performs Grover rotations
to rotate $\ket{N}$ to $\ket{i_*}$, while Algorithm
\ref{alg:entire} first rotates $\ket{N}$ to $\ket{S}$, then rotates
$\ket{S}$ to $\ket{i_*}$.  The final algorithm we consider, the Hybrid Algorithm, first rotates $\ket{N}$
to some superposition of $\ket{N}$ and $\ket{S}$,  and then rotates to $\ket{i_*}.$

\begin{alg}[Hybrid Algorithm] \hfill
\label{alg:hybrid}

 Prepare the state $\ket{N}$ at cost 0. Set $\sop A$ equal to the identity. 
 Since $|\braket{N}{S}|=\sqrt{M/N}$, by Theorem \ref{thm:aa} and
 Corollary \ref{corr:aacost}
there exists an algorithm $\sop B$ that produces a state $\ket{\psi_{\sop B}}$  at cost
\begin{align}
c_S\left\lceil\frac{\left(\arcsin{\sqrt{1-\alpha}} - \arcsin \sqrt{\frac{M}{N}}\right)}{2 \arcsin \sqrt{\frac{M}{N}}}\right\rceil.
\end{align} 
where
\begin{align}
\ket{\psi_{\sop B}}=\sqrt{1-\alpha}\ket{S}+\sqrt{\alpha}\ket{S^\perp}.
\end{align}
By Theorem \ref{thm:aa}, $\ket{S^\perp}$ is a linear combination of $\ket{S}$ and $\ket{N}$
but is orthogonal to $\ket{S}.$ Therefore, $\ket{S^\perp}$ is a superposition of 
all elements not in $S$, and so $\braket{i_*}{S^\perp}=0$. Thus
\begin{align}
\frac{\sqrt{1-\alpha}}{\sqrt{M}}=\braket{\psi_{\sop B}}{i_*}.
\end{align}
Applying Corollary \ref{corr:aacost} again, we can create an algorithm $\sop C$
that has cost
\begin{align}\label{eq:cost_hybrid_pre}
\left\lceil\frac{\arcsin\sqrt{1-\epsilon} - \arcsin \frac{\sqrt{1-\alpha}}{\sqrt{M}}} 
{2\arcsin\frac{\sqrt{1-\alpha}}{\sqrt{M}}}\right\rceil
\left(c_*+2c_S\left\lceil\frac{\left(\arcsin{\sqrt{1-\alpha}} - 
\arcsin \sqrt{\frac{M}{N}}\right)}{2 \arcsin \sqrt{\frac{M}{N}}}\right\rceil\right)
\end{align}
and produces the state $\ket{i_*}$ with probability $1-\epsilon$. In
Appendix \ref{app:cost}, we show there is a choice of $\alpha$
such that,  dropping terms of size at most $O(M^{-1/2})$ or $O((M/N)^{1/4})$ that of the zeroth
order terms, the cost is
\begin{align}\label{eq:fund_cost}
 \textrm{Cost}(\textrm{Hybrid})=\frac{c_S\sqrt{N}\arcsin\sqrt{1-\epsilon} }{2} \sec \left(\phi_{opt}+\sqrt{\frac{M}{N}}\right),
\end{align}
where $\phi_{opt}$ is given by
\begin{align}\label{eq:kcond_orig}
\phi_{opt}=\max
\begin{cases}
0\\
\phi:\tan\left(\phi+\sqrt{\frac{M}{N}}\right)=\phi+\frac{c_*}{c_S}\sqrt{\frac{M}{N}}.
\end{cases}
\end{align}
\end{alg}

When $c_S$ is close to $c_*$, this algorithm approximates Algorithm 1.
When $c_S$ is very small compared to $c_*$, it approximates Algorithm
2. Otherwise, it, in effect, interpolates between the two algorithms.

\section{Lower Bound on Quantum Cost Complexity of STO} \label{sec:lowerbound}

Several techniques give asymptotically tight 
lower bounds on the quantum cost complexity of STO.
We will briefly sketch two approaches for bounding the quantum cost complexity ($QCC)$,
and then discuss a bound on controlled quantum cost complexity ($ConQCC$) in detail. 
The fact that so many approaches give good lower bounds is encouraging;
this means many techniques from (or variations on) the standard query complexity 
toolbox can be applied.

 Our
lower bound on $ConQCC$(STO) is asymptotically
tight with the algorithms of Section \ref{sec:algo}, i.e. Eq. (\ref{eq:cost12}), even though
those algorithms do not use controlled oracles.
Because algorithms that use controlled versions of the oracles are more
powerful than oracles that can not access controlled versions (see Eq. (\ref{eq:compare_costs_def})), this
result proves that not only are our algorithms for STO asymptotically optimal,
but having access to a controlled version of the oracles for STO does
not give an advantage.

When discussing lower bounds on the cost of STO, we will often refer to the
SEARCH problem. We call SEARCH the problem in which one is given a function
$f_*:[N]\rightarrow\{0,1\}$ such that there is exactly zero or one
element $i_*$ such that $f_*(i_*)=1$, and one would like to determine
if there is such an element $i_*; $ in other words, SEARCH is computing OR$(f_*)$    with a promise on $f_*$.

Here are brief descriptions of two methods for lower bounding $QCC$. We describe
them in the context of STO, but they could be applied more generally.

{\textbf{Oracle Simulation:}} Suppose one only has an oracle $\sop
O_*.$ Then one could use this to simulate an oracle $\sop O_S$ by
applying $\sop O_*$, and then subsequently randomly choosing $M-1$ items
to mark. If $M\ll N$, with high probability, the chosen $M-1$ items will not include $\sop O_*$, and this simulated oracle will act
identically to a true $\sop O_S.$ Now any algorithm for STO that uses
this simulated oracle will actually only use $\sop O_*$ to find
the marked item $i_*$, and so the problem reduces to SEARCH.
 Well-known quantum lower bounds on SEARCH \cite{BBBV97} then give a lower bound
on the total number of queries to either $\sop O_*$ or the simulated $\sop O_S$, which
in turn can be used to put a lower bound on the cost. For more details
on oracle simulation, see Section \ref{sec:classical}, in which we use
oracle simulation to bound the classical cost complexity of STO.

{\textbf{Adversary Method:}} One can create an adversary matrix whose
rows and columns are indexed by pairs of oracles $(f_*,f_S)$.
This matrix can be used to create a progress function, and then one can
bound the progress that either oracle $\sop O_*$ or $\sop O_S$ can make. 
This gives lower bounds on the queries needed to $\sop O_*$ and $\sop O_S$ 
to evaluate STO, which in turn can be used to lower bound the cost of STO.
In Appendix \ref{sec:adversary}, we detail how to create this
bound for STO.

\subsection{Lower Bound on Controlled Quantum Cost Complexity of STO}\label{sec:super}

In this section, in order to lower bound $ConQCC$(STO), we consider a new problem in the standard query
model, which we call Expanded Search with Two Oracles (ESTO).  We show
that if we had an algorithm $\sop A$ which could use the control
oracle $C\sop O$ to solve STO with cost $c_\sop A$, then we could
create a new algorithm $\sop A'$ to solve ESTO using $O(c_\sop A)$
queries. We then use the adversary method to lower bound the query
complexity of ESTO, which in turn puts a lower bound on $ConQCC$(STO). This strategy is inspired by
Ambianis's approach for lower bounding the
variable times search problem \cite{A10}.

\begin{figure}[t!]
\centering
\includegraphics[width=13cm]{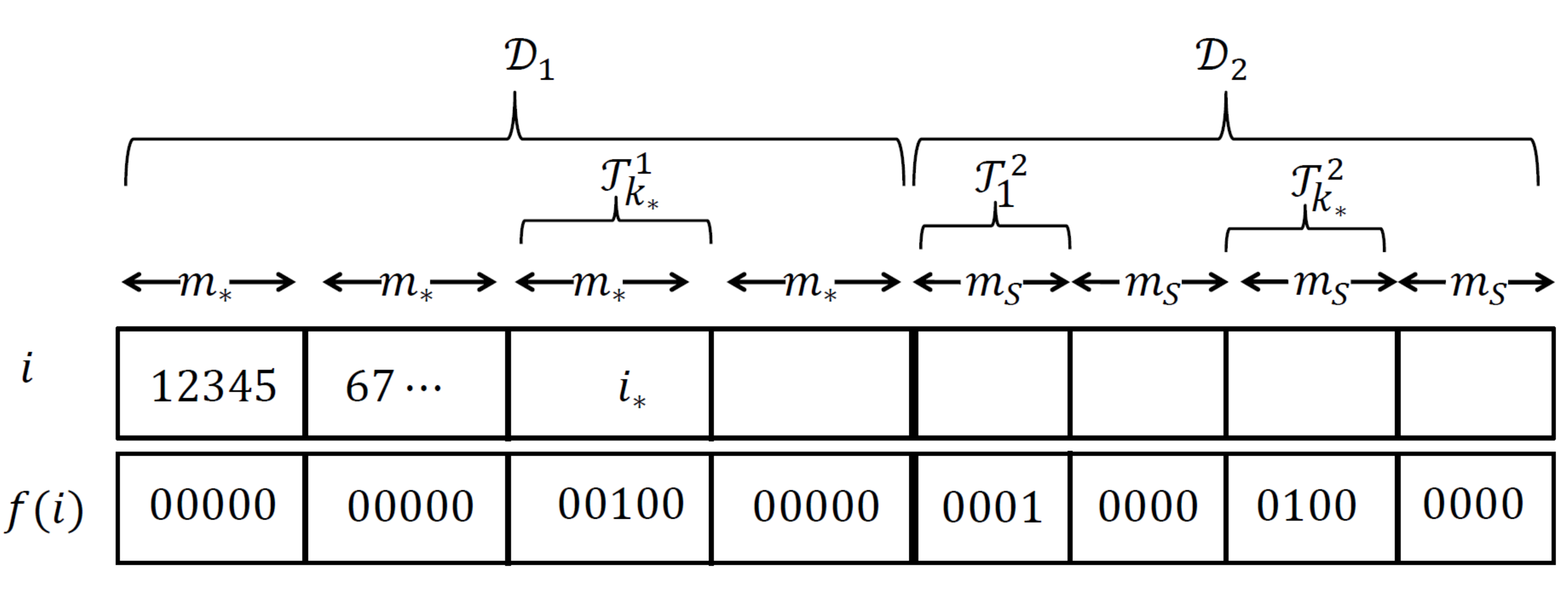}
\caption{\label{plot:ESTO}
A diagram of a function $f$ for which ESTO$(f)=1$. The domain of $f$ is divided
into two parts $\sop D_1$ and $\sop D_2$. Each of these sets are further divided
into $N$ sets of size $m_*$ and $m_S$ respectively. These sets are labeled
$\sop T^1_k$ for sets in $\sop D_1$, and $\sop T^2_k$ for sets in $\sop D_2$. We see there is
exactly one value of $i\in \sop D_1$ with value $1$, and it is in the set $\sop T_{k_*}^1.$
In the case shown in this figure, $S=\{1,k_*\}$, so both $\sop T^2_{k_*}$
and $\sop T^2_{1}$ contain exactly one marked item.}\label{fig:ESTO}
\end{figure}

We first describe the problem ESTO. We suggest referencing Figure
\ref{fig:ESTO} during the description of the problem for a graphical
interpretation. Let $N,$ $M,$ $c_*$ and $c_S$ be as in STO. Without loss
of generality, we can assume $c_*,c_S\gg1$. If they are not, we
can multiply both costs by some large factor $K.$ Then the final
cost is exactly a factor of $K$ larger than it would have been with the original
costs. (If $c_S=0$, this approach does not work, but in that case,
STO reduces to SEARCH). We define
\begin{align}
m_*&=\max\left\{i:\left\lceil \frac{\pi}{4}\sqrt{i}\right\rceil+1\leq c_*,i\in\mathbb Z\right\},\nonumber \\
m_S&=\max\left\{i:\left\lceil \frac{\pi}{4}\sqrt{i}\right\rceil+1\leq c_S,i\in\mathbb Z\right\},\nonumber 
\end{align}

ESTO queries an unknown function $f:[N(m_S+m_*)]\rightarrow \{0,1\}$.  We consider
$\sop D_1=\{1,\dots, Nm_*\}$ to be the ``first part'' of the domain of $f$,
and $\sop D_2=\{Nm_*+1,\dots, N(m_*+m_S)\}$ to be the ``second part'' of the domain. We
further divide $\sop D_1$ ($\sop D_2$) into $N$ blocks of $m_*$ ($m_S$) elements respectively, 
where the elements 
$\sop T^1_k=\{(k-1)m_*+1,\dots,km_*\}$ constitute the $k\tth$ block of $\sop D_1$,
 and the elements $\sop T^2_k=\{Nm_*+(k-1)m_S+1,\dots,
Nm_*+km_S\}$ constitute the $k\tth$ block of $\sop D_2$.

We are promised that there is either exactly zero or one value $i_*\in
\sop D_1$  such that $f(i_*)=1.$  If there is such an $i_*$, we label
the block it is in by $k_*$, so $i_*\in T^1_{k_*}$. Furthermore, if
$i_*$ exists, there is a set $S\in[N]$ such that $|S|=M$, $k_*\in S$,
and for each $k\in S$ there is exactly one value of $i\in\sop T^2_k$
such that $f(i)=1.$  Given such a function $f$, ESTO$(f)=1$ if there
is an item $i_*\in \sop D_1$  such that $f(i_*)=1,$ and $0$ otherwise.

Given an algorithm $\sop A$ for STO that
uses the control oracle $C\sop O$ and has cost $c_\sop A,$ 
we can create an
algorithm $\sop A'$ to solve ESTO that uses $2c_\sop A$ queries.
 Let $y_j^b=1$ for $b\in\{1,2\}$ if
there is an element $i\in\sop T_j^b$ such that $f(i)=1$, and $0$
otherwise.  Then by Claim 2 in \cite{A10}, there is an algorithm $\sop
B$ that takes
$\ket{b,j}\ket{0}\ket{0}\rightarrow\ket{b,j}\ket{y_j^b}\ket{\psi_j^b}$ for some
state $\ket{\psi_j^b}$\ and uses $c_*$ queries if $b=1$ and $c_S$
queries if $b=2$. At the cost of doubling the number
of queries, we can uncompute the final register. Thus there
is an algorithm $\sop B'$ that takes
$\ket{b,j}\ket{0}\rightarrow\ket{b,j}\ket{y_j^b}$
 and uses $2c_*$ queries if $b=1$ and $2c_S$
queries if $b=2$. We also allow for $b=0$, in which case
the algorithm $\sop B'$ applies the identity.

Then we can solve ESTO using our algorithm $\sop A$ for STO. In STO
we are searching for a specific element $i^*\in[N]$ with certain
properties, in ESTO, the search is for a specific block
$k^*\in[N]$ with analogous properties.  We replace an application of the
controlled oracle $C$-$\sop O$ to the state $\ket{b,i}$ with
$b\in\{0,1,2\}$ and $i\in[N]$ with an application of the algorithm
$\sop B'$ to the state $\ket{b,i}$, (which corresponds to searching
the block $\sop T^b_i$, for $b\in\{1,2\}$ and $i\in[N]$, or doing nothing
if $b=0$). The number
of queries required by $\sop B'$ will be twice cost of the equivalent
query made by $\sop A$. Due to the specific structure of $f$, this
algorithm will solve ESTO with a number of queries equal to 
$2c_\sop A.$

Now all that is left is to put a lower bound on the number of queries
needed to solve ESTO. We use Ambainis's adversary bound:

\begin{theorem}[Basic Adversary Bound \cite{A00}]\label{thm:ambainis}
Let $F(f(1),\dots,f(N))$ be a function of $N$ $\{0,1\}$-valued variables
$f(i)$, and let $X,$ $Y$ be two sets of inputs such that $F(f)\neq F(g)$
if $f\in X$ and $g\in Y.$ Let $R\subset X\times Y$ be such that
\begin{itemize}
\item For every $f\in X$, there exist at least $\mu$ different $g\in Y$ such that
$(f,g)\in R.$
\item For every $g\in Y$, there exist at least $\mu'$ different $f\in X$ such that
$(f,g)\in R.$
\item For every $f\in X$ and $i\in[N]$, there are at most
 $l$ different $g\in Y$ such that $(f,g)\in R$ and $f(i)\neq g(i).$
\item For every $g\in Y$and $i\in[N]$, there exist at least $l'$ 
different $f\in X$ such that
$(f,g)\in R$  and $f(i)\neq g(i).$
\end{itemize}
Then, any quantum algorithm computing $F$ with error at most $\epsilon$ on
all valid inputs uses
at least
\begin{align}
\frac{1-2\sqrt{\epsilon(1-\epsilon)}}{2}\sqrt{\frac{\mu\mu'}{ll'}}
\end{align}
queries.

\end{theorem}

For the sets $X$ and $Y$, we only consider functions $f$ where in each block $\sop
T_j^b$, there is at most $1$ marked item. We denote by
$f_{k_*,i_*,S,S'}$ a function where $i_*\in \sop D_1$ is the marked item, 
$k_*$ is the block where the $i_*$ sits
(or $i_*=k_*=0$ if there is no marked item in $\sop D_1$), $S$ is the set of blocks in $\sop D_2$
that have exactly one marked item in each block, and $S'$ is a list
of the $|S|$ items that are marked in the second part of the domain.

Let $X$ be the set of all functions  $f_{k_*,i_*,S,S'}$ with $k_*\neq
0$, $i_*\neq 0$, $|S|=M,$ and $k_*\in S$. From our definition of ESTO, these are
functions for which the algorithm should output 1. Let $Y$ be the set
of functions  $f_{0,0,T,T'}$ with $|T|=M-1.$ Then  $R$ is defined by
$(f_{k_*,i_*,S,S'},f_{0,0,T,T'})\in R$ if and only if
$T\subset S$, $T'\subset S',$ and $k_*\notin T.$  With
this definition of $R$, we have $\mu=1$ while $\mu'=(N-M+1)m_*m_S.$ 
Likewise $l=1$ while $l'=\max\{m_S,m_*\}=m_*$ since $c_*\geq c_S$.
Theorem \ref{thm:ambainis} then gives that the number of queries required to solve ESTO, is at least 
\begin{align}\label{eq:ESTO1}
\frac{1-2\sqrt{\epsilon(1-\epsilon)}}{2}\sqrt{(N-M+1)m_S}.
\end{align}

Eq. (\ref{eq:ESTO1})  does not tell the full story; we can repeat
this procedure with the set $X$ the same as before, but now the set
$Y$ includes all functions $f_{0,0,S,S'}$ such that $|S|=M.$ Then we choose
$(f_{k_*,i_*,S,S'},f_{0,0,T,T'})\in R$ if and only if $T= S$
and $T'= S'.$ With this definition of $R$, we have $\mu=1$, while
$\mu'=Mm_*.$ Likewise $l=1$ while $l'=1$.
 Again using Theorem \ref{thm:ambainis}, we have that the number of 
 queries required to solve ESTO is at least
\begin{align}\label{eq:ESTO2}
\frac{1-2\sqrt{\epsilon(1-\epsilon)}}{2}\sqrt{Mm_*}.
\end{align}

Since $c_*,c_S\gg1$,
we have $m_*=\Omega((c_*)^2)$
and $m_S=\Omega( (c_S)^2)$, so combining Eq. (\ref{eq:ESTO1}) and Eq. (\ref{eq:ESTO2}),
and using the fact that a lower bound on the query complexity of ESTO gives 
a lower bound on the controlled quantum cost complexity of of STO, we have
\begin{align}\label{eq:super_lb}
ConQCC_\epsilon(\textrm{STO})&\ge\frac{1-2\sqrt{\epsilon(1-\epsilon)}}{4}\times
\max\left\{ \sqrt{Mm_*},\sqrt{(N-M+1)m_S}\right\}\\
&=\Omega\left(\max\left\{\sqrt{M}c_*,\sqrt{(N-M+1)}c_S\right\}\right).
\end{align}
With Eq. (\ref{eq:cost12}), this bound proves our
algorithms are asymptotically optimal. In Figure
\ref{fig:compare_algs}, we compare the bound given by the reduction to ESTO
with the Hybrid Algorithm. 
Even though the functions are asymptotically tight, the forms of these two bounds
are quite different.

\begin{figure}[t!]
\centering
\includegraphics[width=10cm]{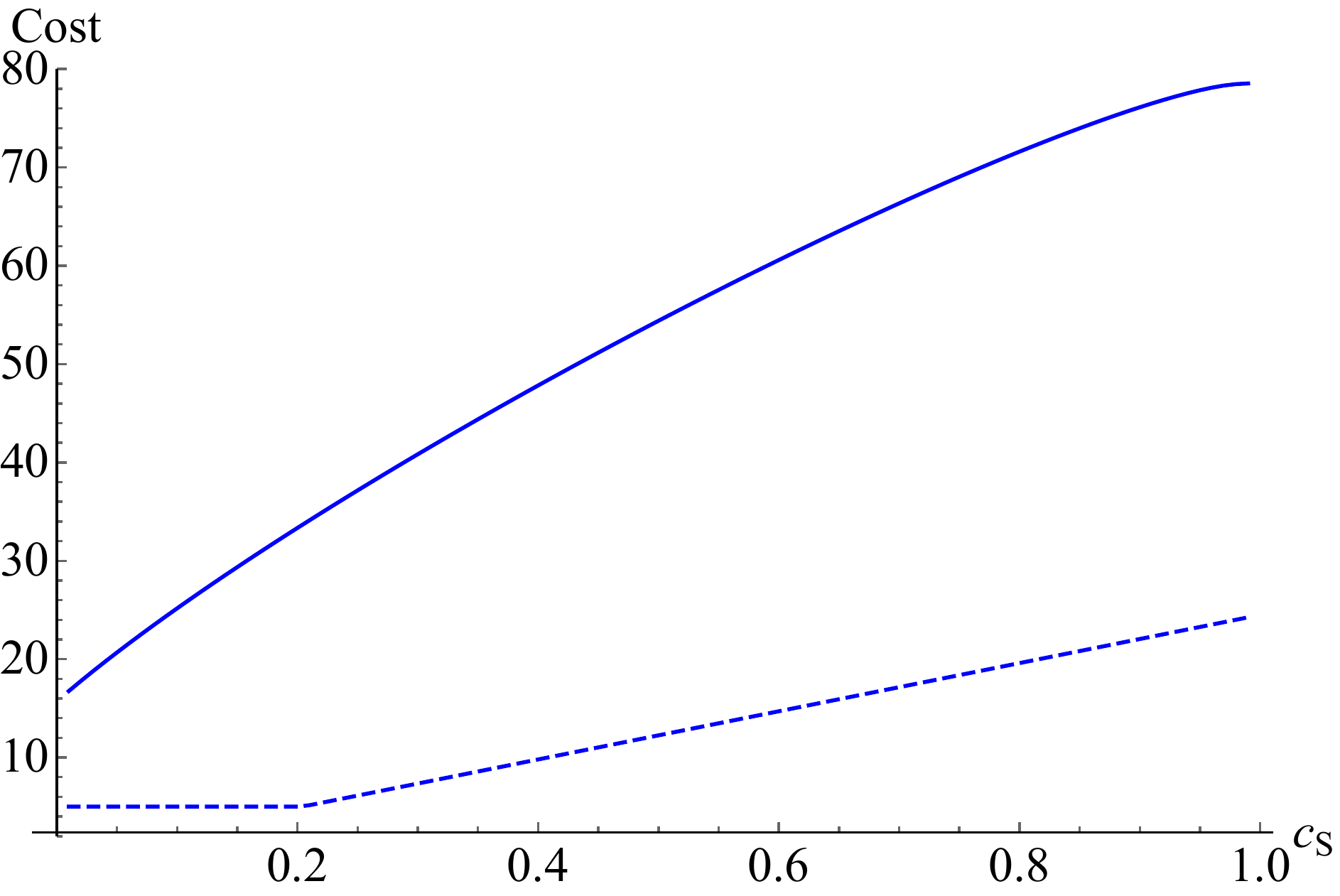}
\caption{\label{plot:compare}
The solid line is the cost of the hybrid algorithm, while the dashed line
is the lower bound on the cost given by Eq. (\ref{eq:super_lb}).
The cost is calculated with $c_*=1$, $N=10^4$, $M=400$ and
$\epsilon=0$ while
$c_S$ is varied.}\label{fig:compare_algs}
\end{figure}


\subsection{Exact Lower Bound for Cost Complexity of STO}\label{sec:geo_not_appendix}

In the introduction, we mentioned several reasons for wanting
to prove exact optimality of our algorithm for STO. Aside from finding
an example besides Grover's algorithm of an exactly optimal algorithm,
proving our algorithm for STO is optimal would have several
other implications. First, the algorithms described in
Section \ref{sec:algo} are all based on amplitude amplification,
so if we can prove these approaches are optimal, that would give evidence
that amplitude amplification is an exactly optimal algorithm for
certain types of unstructured search problems.

Second, if we consider an extension of STO to many oracles,
we can no longer prove asymptotic optimality of our amplitude 
amplification algorithm. Note that in amplitude
amplification, (see Theorem \ref{thm:aa}), 
the inner algorithm ($\sop A$) is applied two times for each application
of the oracle that identifies the target state (if $\sop A=\sop A^{-1}$). This factor
of two is not accounted for in our lower bound of Section \ref{sec:super}. While
this factor of two can be swept under the rug using asymptotic notation, if we 
consider a problem with $k$ nested oracles, and try to apply a similar strategy as for STO
 and use nested amplitude amplification, the innermost algorithm
 will accumulate an extra factor of $2^k$ in the number of times it must be applied.
Using a strategy similar to Section \ref{sec:super} to lower bound
this problem will not catch that factor of $2^k$, for the same reason
the factor of $2$ is not characterized by the oracle simulation and adversary method. In the case of $k=\log N$ nested oracles,
our bounds will no longer be asymptotically tight. Thus, if we can find
an exact bound in the case of STO,
we might be able to extend it to get asymptotically tight bounds
for the case of nested oracles, providing evidence that multiple
nestings of amplitude amplification are optimal for certain
problems.

We have found that proving an exactly tight lower bound for STO is
 a challenge, and in fact we can only prove the hybrid algorithm
is optimal in a limited setting. The difficulty in proving
optimality even in this limited case provides insight into the 
difficulty of the more general case.

The restricted setting we investigate is to only consider \emph{Grover-like}
algorithms. 
\begin{definition}
A \emph{Grover-like} algorithm with oracles $\{\sop O_1,\dots,\sop O_l\}$ that
act on an $N$-dimensional Hilbert space must:
\begin{itemize}
 \item Use only an $N$-dimensional Hilbert space as its workspace,
 \item Initialize in the equal superposition state $\ket{N}=\frac{1}{\sqrt{N}}\sum_{i=1}^{N}\ket{i}$,
 \item Use only the unitaries $\{\sop O_1,\dots,\sop O_l\}$ and $G=\I-2\ketbra{N}{N}$
, and
\item End with a measurement on the standard basis. 
 \end{itemize}
\end{definition} 

If we consider \emph{Grover-like} algorithms for SEARCH, the state of the system is restricted to a
2-dimensional subspace spanned by $\ket{N}$ and $\ket{i_*}.$ Since
$G^2=\sop O_*^2=\I,$ the only possible algorithm is alternating $G$ and
$\sop O_*$, and one can easily track the progress of the state through
the two dimensional space towards $\ket{i_*}$, thus trivially proving that in
this setting, Grover's algorithm is exactly optimal.

We will see in the proof of Theorem \ref{thm:geo} that for STO, the
picture becomes much more complicated. In fact, even in the restricted
setting of \emph{Grover-like} algorithms, we need an additional
assumption to prove optimality. In particular, we show

\begin{restatable}[Exact Lower Bound]{theorem}{geo}  \label{thm:geo}
The cost of every Grover-like algorithm for
STO that succeeds with probability at least $1-\epsilon$ for a constant $\epsilon$ is at least
\begin{align} 
\frac{c_S\sqrt{N}\arcsin\sqrt{1-\epsilon}}{2} \sec \left(\phi_{opt}+\sqrt{M/N}\right),
\end{align} 
where $\phi_{opt}$ satisfies 
\begin{align}\label{eq:kcond}
\phi_{opt}=\max
\begin{cases}
0,\\
\phi:\tan\left(\phi+\sqrt{\frac{M}{N}}\right)=\phi+\frac{c_*}{c_S}\sqrt{\frac{M}{N}}.
\end{cases}
\end{align}
We also require  the conditions $M,N/M\rightarrow \infty$ and $C \rightarrow 0$, where
\begin{align} \label{eq:stupid_condition}
C\equiv\frac{c_S\sqrt{N}}{c_*\sqrt{\epsilon}2M\cos\left(\phi_{opt}+
\sqrt{M/N}\right)}.
\end{align}

\end{restatable}
\noindent Theorem \ref{thm:geo} matches the cost of our hybrid
algorithm, Eq. (\ref{eq:fund_cost}).

The proof of Theorem \ref{thm:geo} can be found in Appendix \ref{app:geo}; here
we provide a very brief sketch. Just as a {\it{Grover-like}} algorithm for standard
search can be thought of as acting on a two dimensional subspace of the full
$N$-dimensional Hilbert space, a Grover-like
algorithm for STO can be thought of as acting on a three-dimensional subspace.
We create a progress function as a position of the state in this subspace
such that $G$ has no affect on the progress function,
while $\sop O_*$ and $\sop O_S$ can cause the progress function to increase or decrease.
We then show that the increase in the progress function due to one of the oracles,
divided by the cost of that oracle, is bounded. In other words, for a given cost,
we can only increase the progress function by a certain amount, no matter which oracle
is used. We finally take the total change in the progress function necessary to achieve
success, and divide by the change in progress per cost to put a lower bound the cost.

\section{Classical Cost Complexity of STO}\label{sec:classical}

In this section, we give bounds on the classical randomized cost
complexity ($RCC$) of STO. We will examine both the exact and bounded
error cost complexity. For the exact cost complexity, we see that
there are two classical algorithms that resemble Algorithm
\ref{alg:direct} and Algorithm \ref{alg:entire}, but whereas in the
quantum case, it is possible to do better with the Hybrid Algorithm,
we prove that there is no classical counterpart to the Hybrid
Algorithm. In the case of exact and bounded error cost complexity, we see a polynomial
increase in cost compared to the quantum case.

In the case of exact classical cost complexity, we have:
\begin{lemma}\label{lemm:classical_zero_error}
The exact (0-error) classical 
cost complexity of STO is
\begin{align}
RCC_0(STO)=\min\{Nc_*,(N-1)c_S+Mc_*\}.
\end{align}
\end{lemma}

\begin{proof}
We consider an adversarial oracle that knows
in advance the queries the algorithm will make. 

Recall that for $i\in[N]$, $f_S$ identifies whether $i\in S$ and $f_*$
identifies whether $i=i_*$. We say an item has been {\it
completely queried} if it has been queried with $f_S$, and is found
to not be an element of $S$, or if it has been queried with $f_*$.
Then the adversarial oracle acts in the following way:
\begin{itemize}
\item The first $M-1$ items that the algorithm
queries using oracle $f_S$ are all elements of $S$.
\item If all elements except one have been queried (but not necessarily completely
queried) using either function $f_*$ or $f_S$,
the final element to be queried will be an element of $S$ (even if this
element is not queried using $f_S$).
\item The last element to be completely queried is the marked item,
if it exists.
\end{itemize}

Any algorithm acting against this adversarial oracle
that makes $q$ queries using $f_S$, has worst-case cost at least
\begin{align}
&Nc_S+Mc_* &\textrm{ if $q=N$},\no
&qc_S+[(N-q)+(M-1)]c_* &\textrm{ if $N-1\geq q\geq M-1$},\no
&qc_S+Nc_* &\textrm{ if $M-1\geq q\geq 0$}.
\end{align}
These expressions are minimized at $q=N-1$ or $q=0$, and we obtain 
\begin{align}
RCC_0(STO)\geq\min\{Nc_*,(N-1)c_S+Mc_*\}.
\end{align}

For the upper bound, consider the following two algorithms.
\begin{alg}\label{alg:classical_all}
Query all items using $f_*$. This algorithm
will find the marked item if it exists with certainty, and has cost $Nc_*$.
\end{alg}
\begin{alg}\label{alg:classical_useM}
 Query all but the last item using $f_S$. Then:
 \begin{itemize}
 \item If $M$ items of $S$ have been found, query $f_*$ on these $M$ items.
 \item If $M-1$ items of $S$ have been found, query $f_*$ on these $M-1$ items, and also the last item (the item that was not queried using $f_S$).
 \item Otherwise $|S| \neq M$ and therefore no marked item exists.
 \end{itemize}
This algorithm will find the marked item if it exists with certainty, and has cost $(N-1)c_S+Mc_*$.
\end{alg}
Thus we have
\begin{align}
RCC_0(STO)\leq\min\{Nc_*,(N-1)c_S+Mc_*\}.
\end{align}
\end{proof}

Algorithm \ref{alg:direct} can be thought of as the quantum version of
Algorithm \ref{alg:classical_all}, while Algorithm \ref{alg:entire}
can be thought of as the quantum version of Algorithm \ref{alg:classical_useM}.
In the $0$-error classical case, these two approaches tell the whole
story. However, in the quantum case, you can do better with the Hybrid Algorithm.
The Hybrid Algorithm works by doing something very quantum, which is to partially
search for the elements of $S.$ In the classical case, this doesn't work. Once
you've found an element of $S$, you've found it; there is no way to partially
find an element of $S$.

With Lemma \ref{lemm:classical_zero_error}, we've proven that in the $0$-error
case, we can obtain a polynomial reduction in cost by using a quantum algorithm
for $STO.$ Next, we show this polynomial reduction holds even in the case of 
bounded error algorithms. We do this by reducing $STO$ to the problem 
of SEARCH. Recall that for SEARCH, we have:
\begin{lemma}\label{lem:classical_search}
Any randomized classical algorithm that solves SEARCH with bounded probability 
must query $f_*$ at least 
$\Omega(N)$ times.
\end{lemma}

Now we can prove the reduction of STO to standard search:
\begin{lemma}
Any randomized classical algorithm that solves STO with bounded probability
of error must use as least $\Omega(N)$ queries to either $f_*$ or $f_S$,
as long as $M/N\leq 1/9.$
\end{lemma}
\begin{proof}
Suppose there is a randomized algorithm $\sop A$ that solves
STO with probability $3/4$ and
makes $q_*$ queries to $f_*$ and $q_S$ queries to $f_S$. Then we will 
use $\sop A$ to find $i_*$ in the case when we are given $f_*$ but not $f_S.$
To do this, we will use $f_*$ to create a function that behaves similarly
to $f_S.$ We choose a subset $T\in[N]$
with $|T|=M-1$ at random, and create a function $f_T$ that acts as
\begin{align}
 f_T(i)=
\begin{cases}
1 &\text{ if }i\in T\\
0 &\text{ if }i\notin T.
\end{cases}
\end{align}
Then we create the function $ \tilde{f}_S$ to simulate $f_S$, where
\begin{align}
 \tilde{f}_S(i)=f_T(i)\vee f_*(i).
 \end{align} 
Each time we want to query $\tilde{f}_S$, we must query $f_*(i)$.
Notice that $\tilde{f}_S$ behaves like a valid $f_S$ function unless
$i_* $ exists and $i_*\in T$ (because in this case $\tilde{f}_S$ marks
$M-1$ items instead of $M$.) $i_*\in T$ with probability
$\frac{M-1}{N.}$

We create $\tilde{f}_S$ as above, and we implement $\sop A$, but every time $\sop A$ 
asks us to apply $f_S$, we instead apply $\tilde{f}_S$. This new
algorithm will succeed with probability $3/4(1-(M-1)/N)\geq 2/3$,
because it succeeds with probability $3/4$ as long as $i_*\notin \sop F$. 
This means we have created an algorithm for 
standard search which uses 
$q_*+q_S$ queries to $f_*$ and which succeeds with probability 2/3. But 
by Lemma \ref{lem:classical_search}, we must have $q_*+q_S=\Omega(N).$
\end{proof}

Finally, we note that there is an additional restriction on the number of
queries to $f_*$:
\begin{lemma}\label{lem:Msearch}
Any randomized classical algorithm that solves STO with bounded probability must use
at least $\Omega(M)$ queries to $ f_*$.
\end{lemma}
\begin{proof}
Suppose the elements of the subset $S$
were known. Then in the worst case, that would still only narrow down the search to $M$ items. (This
is the worst case because if $|S|\neq M$, then one immediately knows there is no marked item.) One
must then perform a search for one marked item out of $M$, which requires $\Omega(M)$ 
queries via Lemma \ref{lem:classical_search}.
\end{proof}

Now we can state our lower bound on the query cost of STO:
\begin{theorem}
The bounded error classical randomized cost complexity of STO is
\begin{align}\label{eq:lb1}
RCC(STO)=\min&\left\{\Omega(c_SN+c_*M),\Omega(c_*N)\right\}.
\end{align}
\end{theorem}

\begin{proof}
When $M/N\leq 1/9$, we solve the following linear program:
\begin{align}
\textrm{minimize: } &q_*c_*+q_\sop S c_\sop S\no
\textrm{subject to: }& q_*\geq f_1(M,\epsilon)\no
&q_*+q_\sop S \geq f_2(N,M,\epsilon).
\end{align}

When $M/N> 1/9$, from Lemma \ref{lem:Msearch}, we have have $q_*=\Omega(M)=\Omega(N)$, so the
cost is as least $\Omega(c_*M)=\Omega(c_*N).$
\end{proof}

Comparing Eq. (\ref{eq:lb1}) with Eq. (\ref{eq:cost12}), we see that there is always 
a separation between the quantum and classical costs of STO. In particular,
to get the quantum scaling from the classical scaling, simply replace all 
$M$'s and $N$'s by $\sqrt{M}$ and $\sqrt{N}$.

\section{Conclusions and Open Questions}
While query complexity is a well understood and powerful tool for quantifying
the power of quantum computers, there are still problems that are not easily
characterized by query complexity. Cost complexity is one way of extending
the standard query model, and we've argued that this approach has potential applications in 
constraint satisfaction problems. 

While we motivated STO with problems like $k$-SAT, graph isomorphism,
and the traveling salesman problem,
it is not obvious how much of a speed-up an STO inspired
algorithm for these problems would be. The speed-up in STO depends critically
on $N,$ $M,$ $c_*,$ and $c_s$. It would be interesting
to calculate approximately what this relationship is, for example,  in
a random $k$-SAT instance. Once this relationship is better
understood, we could determine the amount of speed-up
an STO algorithm would give for such a problem.
However, even with a better understanding of this relationship,
it is unlikely that $M$ would be known exactly. In that case, a method
such as fixed point search \cite{YLC14} might be helpful.

STO is a very simple extension of a search
problem, and thus the methods described here all
have a Grover-ish flavor to them. It would be interesting to 
find well motivated problems for the cost complexity model where other
quantum algorithms could be employed.

We have also left open the question of the exact cost of STO. We believe
our algorithm is optimal, but it seems new techniques
are needed to prove it.

\section{Acknowledgments}
The authors would like to thank Edward Farhi, Andrew Childs, and Aram Harrow for
illuminating discussions. Funding for SK provided by the Department of Defense. HSL and CYL are supported by the ARO grant Contract Number W911NF-12-0486. 
CYL gratefully acknowledges support from the Natural Sciences and Engineering Research Council of Canada.

\bibliography{MObib}
\bibliographystyle{plain}

\appendix

\section{Analysis of the Hybrid Algorithm}\label{app:cost}
Throughout this section, when we are calculating something ``to zeroth order'', we drop 
terms whose sizes are $O(M^{-1/2})$ or $O((M/N)^{1/4})$ multiplied by the size of the largest term.

In Section \ref{sec:algo}, Eq. (\ref{eq:cost_hybrid_pre}), we showed that the cost of the Hybrid Algorithm is
\begin{align}\label{eq:orig_cost}
 \textrm{Cost}(\textrm{Hybrid})=&\left\lceil\frac{\arcsin\sqrt{1-\epsilon} - \arcsin \frac{\sqrt{1-\alpha}}{\sqrt{M}}} {2\arcsin\frac{\sqrt{1-\alpha}}{\sqrt{M}}}\right\rceil
\nonumber\\
&\times\left(c_*+2c_S\left\lceil\frac{\left(\arcsin{\sqrt{1-\alpha}} - \arcsin \sqrt{\frac{M}{N}}\right)}{2 \arcsin \sqrt{\frac{M}{N}}}\right\rceil\right).
\end{align}
In this appendix, we prove that in the limit of  $M\rightarrow \infty$ and $N/M\rightarrow \infty,$ there is a choice of $\alpha$ such
that the cost is
\begin{align}\label{eq:fund_cost_app}
 \textrm{Cost}(\textrm{Hybrid})=\frac{c_S\sqrt{N}\arcsin\sqrt{1-\epsilon} }{2} \sec \left(\phi_{opt}+\sqrt{\frac{M}{N}}\right),
\end{align}
where $\phi_{opt}$ is given by
\begin{align}\label{eq:kcond_app}
\phi_{opt}=\max
\begin{cases}
0\\
\phi:\tan\left(\phi+\sqrt{\frac{M}{N}}\right)=\phi+\frac{c_*}{c_S}\sqrt{\frac{M}{N}}.
\end{cases}
\end{align}

We first define
\begin{align}
t=\left\lceil\frac{\left(\arcsin{\sqrt{1-\alpha}} - \arcsin \sqrt{\frac{M}{N}}\right)}{2 \arcsin \sqrt{\frac{M}{N}}}\right\rceil,
\end{align}
so $t$ is a non-negative integer.
Substituting $t$ for $\alpha$ in Eq. (\ref{eq:orig_cost}), we obtain
\begin{align}
 \textrm{Cost}(\textrm{Hybrid}) =&  (2tc_S + c_*)  \nonumber\\
 &\times\left\lceil\arcsin\sqrt{1-\epsilon}
 \left[2\arcsin \left( \frac{\sin\left((2t+1) \arcsin \sqrt{\frac{M}{N}}\right)} {\sqrt{M}} \right)\right]^{-1} -1/2\right\rceil.
\end{align}
To zeroth order, this becomes
\begin{align}
\textrm{Cost}(\textrm{Hybrid}) &=  \frac{(2tc_S + c_*)  \sqrt{M}\arcsin\sqrt{1-\epsilon}}{
2\sin\left((2t+1)\sqrt{\frac{M}{N}} \right)}.
\end{align}
Finally, we denote $\phi=2t\sqrt{M/N}$
to obtain
\begin{align}\label{eq:cost_plugin}
\textrm{Cost}(\textrm{Hybrid}) &=
\frac{\left(\phi c_S + \sqrt{\frac{M}{N}}c_*\right)  
\sqrt{N}\arcsin\sqrt{1-\epsilon}}{
2\sin\left(\phi+\sqrt{\frac{M}{N}} \right)}.
\end{align}

We take the partial
derivative of the cost with respect to $\phi $, and set it to zero
to find the value of $\phi$ that gives the smallest cost. We find the cost is minimized
when $\phi=\phi_{opt},$ where $\phi_{opt}$ satisfies
\begin{align}\label{eq:phimin}
\tan\left(\phi_{opt}+\sqrt{M/N}\right) &= \phi_{opt} + \frac{c_*}{c}\sqrt{M/N}.
\end{align}
Notice that there is always a solution with $\phi_{opt}\in[-\sqrt{M/N},\pi/2]$. However $t$ is 
non-negative, so if $\phi_{opt}<0$ we set $\phi_{opt}=0$. This condition,
along with Eq. (\ref{eq:cost_plugin}) and Eq. (\ref{eq:phimin}), immediately gives the cost claimed in 
Eq. (\ref{eq:fund_cost_app}).

We might not be able to exactly attain this cost, because
$t$ must be an integer, so we might only be able to set $\phi$ close
to $\phi_{opt}$. We show that even if we can't set $\phi$
exactly to $\phi_{opt}$, we can still attain the cost of Eq. (\ref{eq:fund_cost_app}),
to zeroth order.

There are two cases to consider.
 In the first case, we assume $(M/N)^{1/4}\leq\phi_{opt}\leq \pi/2.$
 We require that $t$ be a non-negative integer, so we choose 
$
t=\left\lceil(\phi_{opt}\sqrt{N})/(2\sqrt{M})\right\rceil,
$
and hence we set
\begin{align}
\phi=\left\lceil\frac{\phi_{opt}}{2}\sqrt{\frac{N}{M}}\right\rceil2\sqrt{\frac{M}{N}}.
\end{align}
For that choice, notice that 
\begin{align}\label{eq:phi_to_phiopt}
\phi-\phi_{opt}=O\left((M/N)^{1/2}\right).
\end{align} 
This allows us to relate terms involving $\phi$ to those involving $\phi_0$:
\begin{align}\label{eq:approx_subs}
\sin\left(\phi+\sqrt{M/N}\right)=& \sin\left(\phi_{opt}+\sqrt{M/N}\right)\pm O((M/N)^{1/2})\no
=&\sin\left(\phi_{opt}+\sqrt{M/N}\right)\left(1\pm O((M/N)^{1/4})\right)\no
=&\left(\phi_{opt} + \frac{c_*}{c_S}\sqrt{M/N}\right)\cos\left(\phi_{opt}+\sqrt{M/N}\right)\left(1\pm O\left((M/N)^{1/4}\right)\right)\no
=&\left(\phi + \frac{c_*}{c_S}\sqrt{M/N}\right)\cos\left(\phi_{opt}+\sqrt{M/N}\right)\left(1\pm O\left((M/N)^{1/4}\right)\right),
\end{align}
where in the first line, we use the angle addition formula and Eq. (\ref{eq:phi_to_phiopt}); 
in the second,
we use the assumption that $\phi_{opt}\geq (M/N)^{1/4}$; in the
third line we applied Eq. (\ref{eq:phimin}); and in the last we have used Eq. (\ref{eq:phi_to_phiopt})
and the assumption on the size of $\phi_{opt}$.
Plugging Eq. (\ref{eq:approx_subs}) into our expression 
for the cost in Eq. (\ref{eq:cost_plugin}), we have that to
zeroth order, we obtain 
Eq. (\ref{eq:fund_cost_app}), as desired.

We now consider the second case, when $0\leq \phi_{opt}<(M/N)^{1/4}$.
In this case, we simply set $t=0,$ and hence $\phi=0.$ Plugging $\phi=0$ the cost
of Eq. (\ref{eq:cost_plugin}), we have, 
to zeroth order,
\begin{align}\label{eq:cost0}
\textrm{Cost}(\textrm{Hybrid}) &=  \frac{\arcsin\sqrt{1-\epsilon}\sqrt{N}}{2} c_*.
\end{align}
We will show that Eq. (\ref{eq:cost0}) and Eq. (\ref{eq:fund_cost_app}) are
equivalent for $0\leq \phi_{opt}<(M/N)^{1/4}$. We have 
\begin{align}\label{eq:expandsec}
\sec\left(\phi_{opt}+\sqrt{M/N}\right)=1+O\left((M/N)^{1/4}\right).
\end{align}
We can expand Eq. (\ref{eq:phimin}) to get
\begin{align}\label{eq:ccstar}
c_S=c_*\left(1-O\left((M/N)^{1/4})\right)\right).
\end{align}
Plugging Eqs. (\ref{eq:expandsec}) and (\ref{eq:ccstar})  into Eq. (\ref{eq:fund_cost_app}) and keeping
only zeroth order terms, we recover
Eq. (\ref{eq:cost0}).

\section{Proof of Theorem \ref{thm:geo}}\label{app:geo}
In this section, we prove the following theorem:

\geo*

\begin{proof} 
Throughout this section, when we say to zeroth order, we mean dropping terms of size at most $O(M^{-1/2})$ or $O\left((M/N)^{1/2}\right)$ or $O(C)$ of the zeroth
order terms.

Since we only consider the operations $\sop O_S$, $\sop O_*$, and $G$,
the state of the system never leaves the three-dimensional space
spanned by the orthonormal states
\begin{align}
\left\{
\begin{array}{ccc}
\ket {i_*},&
\ket{S^-} = \frac{1}{\sqrt{M-1}}\sum_{i \in S-\{i_*\}}
\ket i,&
\ket {S^\perp} = \frac{1}{\sqrt{N-M}} \sum_{i\notin S} \ket i
\end{array}
\right\}.
\label{eq:standard_bases}
\end{align}

It turns out that it is more convenient to work in a slightly shifted basis from 
that of Eq. (\ref{eq:standard_bases}). We instead use the orthonormal basis states:
\begin{align}
\ket{x}&= \cos\theta_0\ket{i_*} - \sin\theta_0\ket{S^-}, \nonumber\\
\ket{y}  &= \cos\phi_0\sin\theta_0 \ket{i_*}+\cos\phi_0\cos\theta_0\ket{S^-} - \sin\phi_0 \ket{S^\perp}, \nonumber\\
\ket{z} &=\sin\phi_0\sin\theta_0\ket{i_*}+\cos\theta_0\sin\phi_0\ket{S^-}+\cos\phi_0\ket{S^\perp} \nonumber\\
&= \ket{N}.
\end{align}
We can think of these states as forming the axes of a 3-dimensional space,
where a state 
\begin{align}
\ket{\chi}=x\ket{x}+y\ket{y}+z\ket{z}
\end{align}
is identified with the point $(x,y,z)$.
Then if the algorithm
is initialized in the equal superposition state $\ket{z}$, the goal of the 
algorithm is to move from the $\ket{z}$-axis towards the $\ket{x}$-axis.

Since any normalized state of the system corresponds to a point on the unit sphere in this space, 
let us now introduce polar coordinates, with the $\ket{x}$-axis as
the polar axis. Specifically, we associate the state $\ket{\chi}$
with the polar coordinates $(\theta, \phi)$, where
\begin{equation} 
x = \sin\theta, \quad y = \cos \theta \sin \phi, \quad z =
\cos\theta\cos\phi 
\end{equation} 
for
$\theta \in [-\pi/2,\pi/2].$ (The variable $\phi$ in this section plays a nearly
identical role to $\phi$ in Appendix \ref{app:cost}, so we use the same variable name.)

If we multiply a state by $-1$, this transforms the coordinates  
from $(\theta,\phi)$ to
$(-\theta,\phi+\pi)$.
Because overall phases do not affect the state, we can apply this transformation 
for free. In particular, we use it to 
``pick a gauge'' and choose the coordinates
that satisfy $\theta\geq0$.



For a Grover-like algorithm which finds the marked state with high success
probability, the algorithm starts at the point $(\theta = 0, \phi = 0)$,
and must end near $\theta=\pi/2$. We define a progress function $H(\theta, \phi)$, for
$\theta>0$, as
\begin{equation} \label{eq:cost_def}
H(\theta,\phi)= \theta - k    
 \min\limits_{\ell \in \mathbb{Z}}  |\phi + 2\ell\pi -
\pi/2|,
\end{equation} 
where 
\begin{align}\label{eq:kdeg}
k&=\theta_0\cos(\phi_{opt}+\phi_0),\\
\label{eq:phi0_def}
\phi_0&=\arcsin\sqrt{M/N},\\
\label{eq:theta0_def}
\theta_0&=\arcsin\sqrt{1/M},\\
\label{eq:kcond_f}
\phi_{opt}&=\max
\begin{cases}
0\\
\phi:\tan(\phi+\phi_0)=\phi+\frac{c_*}{c}\phi_0.
\end{cases}
\end{align}
The second term of $H(\theta, \phi)$ is proportional to the angular distance of
$\phi$ to $\pi/2$ (taken so the distance is $< \pi$).

Before we analyze how each unitary changes the progress function, we
will look at the total progress that must occur for the algorithm to
succeed. The total progress gained by the algorithm must be larger
than the difference between the value of the progress function at the starting point and
the end point. We pick the starting point as the \emph{last} time the
algorithm increases $\theta$ from less than $2 \theta_0$ to more than
$2\theta_0$, and $\phi\geq0$. (We require $\phi\geq 0$ for Lemma \ref{lemm:neg_phi},
and we require $\theta\geq 2\theta_0$ in order to calculate the progress
due to $\sop O_*.$) We will show later that such a point
will always exist for any successful algorithm, and also that at such
a point $\theta<6\theta_0.$ Thus the value of the progress function at the starting
point is at most $6\theta_0$.

For the end point of the algorithm, note that the probability of success is 
\begin{align}\label{eq:epsilon_bound}
\sin^2(\theta)>1- \epsilon,
\end{align}
to zeroth order. 
Thus the total change in progress function
is at least
\begin{align}\label{eq:bound_cos_theta}
\arcsin\sqrt{1-\epsilon}-k \pi -6\theta_0>\arcsin\sqrt{1-\epsilon}-(6+ \pi)\theta_0,
\end{align} 
where we bound $k$ using Eq. (\ref{eq:kdeg}), and the $k\pi$ term comes
from the worst possible value of $\phi$ when $\theta$ gets sufficiently large.

We note the following: from Eq. (\ref{eq:fund_cost}) and  Eq. (\ref{eq:kdeg}) we see that the 
cost of the optimal algorithm is at most 
\begin{align}\label{eq:drop_terms}
\frac{c_S\arcsin\sqrt{1-\epsilon}}{\phi_0 k},
\end{align}
and from Eq. (\ref{eq:bound_cos_theta}) the change in the progress function is at least $\arcsin\sqrt{1-\epsilon}-(6+\pi)\theta_0$; 
therefore the progress per unit cost must be at least $\phi_0 k/c_S$,
to zeroth order. 
It therefore follows that when calculating the change in progress
function, we only need to keep track of terms up to order $O(\phi_0 k
/ c_S)$ per cost. For example, for $\sop O_*$, we need only keep
track of the change in {\it{progress}}  (not progress per cost)
up to order $O(\phi_0 k c_*/c_S)$.

The change in the progress function $H(\theta,\phi)$ due to the unitaries $G$, $\sop O_S$, and $\sop O_*$
can be calculated by how they change the coordinates $(\theta,\phi)$
of a state.
After some algebra and using our gauge choice, we obtain  
\begin{itemize}
\item $G$: The unitary $G$ is a reflection about the $z$-axis, 
and in polar coordinates is the map 
\begin{align}\label{eq:G_angle_change}
G:\: (\theta,\phi) &\rightarrow 
(\theta,\pi-\phi).
\end{align}
Comparing with Eq. (\ref{eq:cost_def}), we see $G$ has no effect on the progress function.

\item $\sop O_S$: The oracle $\sop O_S$ 
is a reflection about the state  
with has polar coordinates $(\theta=0,\phi = -\phi_0)$.
\begin{align}\label{eq:S_angle_change}
\sop O_S:\: (\theta,\phi) &\rightarrow 
(\theta, \pi - \phi - 2\phi_0)
\end{align}
We see that $\sop O_S$ can change the progress function by at most
$2\phi_0k$. Thus the increase in the progress function per cost due
to $\sop O_S$ is at 
most
\begin{align}\label{eq:OS_increase}
\frac{2\phi_0k}{c_S} = \frac{2\phi_0\theta_0\cos(\phi_{opt}+\phi_0)}{c_S}.
\end{align}

\item $\sop O_*$: The oracle $\sop O_*$ 
is a reflection about
the state $\ket{i_*}$, which is close to $\ket{x}$. We find $\sop O_*$
transforms coordinates as
\begin{align}
\theta &\rightarrow \theta + 2\theta_0\sin(\phi+\phi_0) + O(\theta_0^2) \\
\phi &\rightarrow 
\pi +\phi+O\left(\frac{\theta_0}{\cos\theta}\right).\label{eq:ostar_phi_change}
\end{align}
\end{itemize}

Now we consider how $\sop O_*$ affects the progress function; unlike the previous cases, 
which we calculated exactly, we will only analyze this case to zeroth order.
We will first show that we can assume $|\phi| \le \pi/2$. 
Suppose that $|\phi| > \pi/2$ just before we would like to apply $\sop O_*$. Then
instead of applying $\sop O_*$, we apply $G\sop O_* G$. One can check
that with this replacement, when $\sop O_*$ is applied, $|\phi|\le\pi/2.$
%
Furthermore one can verify that this replacement causes $\theta$ to increase (which can only be good for the 
progress function), while on the other hand, the value of $\phi$
changes by at most $O\left(\theta_0/\cos\theta\right)$ due to this
replacement, resulting in a change in the progress function of size
$O\left(k\theta_0/\sqrt{\epsilon}\right)$ (using Eq. (\ref{eq:epsilon_bound}) 
to bound $\cos\theta)$. 
Using our assumption that that $C = o(1)$, this change has order less 
than $O(\phi_0 k c_*/c_S)$, and so can be discarded using the argument
following Eq. (\ref{eq:drop_terms}). 
We can therefore assume that $\sop O_*$ is always applied at $|\phi|\leq \pi/2.$

Now we can examine the change in the progress function due 
to the action of $\sop O_*$. The increase in the progress function is
\begin{align}
& 2\theta_0\sin(\phi + \phi_0) + O\left(\theta_0^2\right) \no
& - k\left(\min\limits_{\ell \in \mathbb{Z}} 
 |-\phi +2\ell\pi - \pi/2| - \min\limits_{\ell \in \mathbb{Z}}  |\phi + 2\ell\pi - \pi/2|\right) + O\left(\frac{k\theta_0}{\cos\theta}\right).
 \end{align}
 Since $|\phi|\leq \pi/2,$ the  increase in the progress function due to $\sop O_*$ is less than
\begin{align}\label{eq:cost_with_k}
 2\theta_0\sin(\phi+\phi_0) - 2\phi\theta_0\cos(\phi_{opt}+\phi_0) +O\left(\frac{\theta_0^2}{\sqrt{\epsilon}}\right),
\end{align}
where we have used the value of $k$ from Eq. (\ref{eq:kdeg}) and bounded $\cos\theta$ with Eq. (\ref{eq:epsilon_bound}).

Taking the first and second derivatives of Eq. (\ref{eq:cost_with_k})
with respect to $\phi$, we see
that when $\phi\geq 0$, the increase in the
progress function is maximized when $\phi=\phi_{opt}$.  It turns out
that if one applies $\sop O_*$ at $\phi<0$, it is sometimes possible
to achieve a larger increase in progress per cost than when $\phi\ge 0$. However,
we show at the
end of this section, (Lemma \ref{lemm:neg_phi}), that applying $\sop
O_*$ when $\phi < 0$ will always be less efficient (up to higher
order terms) in terms of the increase in progress function per cost, than
applying $\sop O_*$ at $\phi = \phi_{opt}$, when viewed in the context of the
larger algorithm. Applying the definition of
$\phi_{opt}$  from Eq. (\ref{eq:kcond_f}) to Eq. (\ref{eq:cost_with_k}), and using
the definition of $C$ from Eq. (\ref{eq:stupid_condition}),
the increase in the progress
function  due to $\sop O_*$ is less than
\begin{align}\label{eq:increase}
\frac{c_*2\phi_0\theta_0\cos(\phi_{opt}+\phi_0)}{c_S}
\left(1+O(\phi_0^2)+O\left(C\right)\right),
\end{align}
where the $O(\phi_0^2)$ term accounts for the case that $\phi_{opt}=0$.

From Eq. (\ref{eq:OS_increase}) and Eq. (\ref{eq:increase})
we see that (to zeroth order) the maximum increase in the progress function per cost is the same
whether $\sop O_*$ is applied or $\sop O_S$ is applied.
Dividing the total necessary change in progress (Eq. (\ref{eq:bound_cos_theta})) by the maximum change in progress per cost
(Eq. (\ref{eq:increase}))
gives us the minimum cost:
\begin{align}
\arcsin\sqrt{1-\epsilon}
\frac{c_S}{2\phi_0\theta_0\cos(\phi_{opt}+ \phi_0)}\left(1-O(C)-O(M^{-1/2})-O\left((M/N)^{-1/2}\right)\right).
\end{align}
In the limit of $N,M\rightarrow\infty$ and $C\rightarrow 0$, (to zeroth order) we have that the cost is
at least
\begin{align}
\arcsin\sqrt{1-\epsilon}
\frac{c_S\sqrt{M}}{2\phi_0\cos(\phi_{opt}+ \phi_0)},
\end{align}
which matches the cost of Eq. (\ref{eq:fund_cost}).

We now justify why the value of the progress function must be less than $6\theta_0$
 when we start tracking it. Immediately before we start tracking the progress function,
we have $\theta<2\theta_0,$ so the bound on the increase in progress given by Eq. (\ref{eq:cost_with_k})
does not necessarily apply. However, it is simple to show that the increase in the progress
function due to $\sop O_*$ is always bounded by
$2\theta_0,$
where we have dropped terms of $O(\theta_0^2/\sqrt{\epsilon})$ as before. 
Thus if $\theta<2\theta_0$, and then $\sop O_*$ is applied, $\theta$ can increase by at most
$2\theta_0$, and so the new value of $\theta$ satisfies $\theta<4\theta_0.$ At this point,
$\theta>2\theta_0,$ but $\phi$ might be negative. Notice that $\theta$ can not increase
unless $\sop O_*$ is applied, (and $\theta$ must increase in order to obtain a high
probability of success) but $\sop O_*$ flips the sign of $\phi,$ so after applying
$\sop O_*$ at most one more time, we will have both the conditions $\theta>2\theta_0$
and $\phi\geq0$ satisfied, at which point we start tracking the progress function. This tells
us that the value of $\theta$ will be at most $6\theta_0$ when we start tracking the progress
function.
\end{proof}

\begin{lemma}\label{lemm:neg_phi}
Suppose there is an algorithm than applies $\sop O_*$ when $\phi<0.$ Then there is always an
alternative algorithm that achieves the same or greater increase in progress for the same or less
cost (up to zeroth order), but applies $\sop O_*$ only when $\phi\geq0.$
\end{lemma}

\begin{proof}
We begin by classifying the the possible sequences of $\sop O_*$, 
$\sop O_S$,  and $G$ the algorithm can take. We will use notation such that unitaries
act from right to left, so $G\sop O_*$ signifies $\sop O_*$ acts first, and then
$G$ acts.

First look at $\sop O_*$. We can always assume $\sop O_*$ is followed
by a $G$; if it is not, insert a $GG$ pair after the $\sop O_*$. Note
in the discussion following Eq. (\ref{eq:ostar_phi_change}), we proved
that we can assume $|\phi|<\pi/2$ before applying $\sop O_*$. With Eqs.
(\ref{eq:G_angle_change}) and (\ref{eq:ostar_phi_change}) we have
\begin{align}
 G\sop O_*: \phi\rightarrow -\phi+O\left(\frac{\theta_0}{\cos\theta}\right).
\end{align}
Since $|\phi|< \frac{\pi}{2}$ before $G\sop O_*$ acts, we also have $|\phi|< \frac{\pi}{2}$ after $G\sop O_*$ acts, up to an additive
factor of $O\left(\frac{\theta_0}{\cos\theta}\right)$, which we can
ignore thanks to the discussion following Eq. (\ref{eq:drop_terms}). Therefore $G \sop O_*$ maps $\phi$ inside the $|\phi|< \frac{\pi}{2}$ region.

In between applications of $G \sop O_*,$ 
there is always a sequence of one of the following forms:
\begin{align} 
 (G\sop O_S)^m ,  \hspace{.5cm} 
 G(G\sop O_S)^m, \hspace{.5cm}
 (\sop O_SG)^m, \hspace{.5cm}   \text{ or } 
 \hspace{.5cm} G(\sop O_SG)^m,\label{eq:possible_sequences}
 \end{align}
 where $m$ is a non-negative integer
 that indicates multiple applications of the unitary sequence inside the parenthesis.
These are the only possible sequences because $\sop O_S\sop O_S=I$ and $GG=I$.
Combining  the action of $G$ and $\sop O_S$ in Eqs.  (\ref{eq:G_angle_change}) and  (\ref{eq:S_angle_change}) we get 
\begin{align}
&(\sop O_SG)^m:(\theta,\phi)\rightarrow (\theta, \phi-2m\phi_0) \\
&(G\sop O_S)^m:(\theta,\phi)\rightarrow (\theta, \phi+2m\phi_0).
\end{align}
Thus the 4 sequences of Eq. (\ref{eq:possible_sequences}) rotate $\phi$ by some amount $\pm 2m\phi_0$,
possibly followed by the transformation $\phi\rightarrow \pi-\phi$.

Now we focus on the algorithm's action on $\phi$. Since the $G \sop O_*$'s are mapping $\phi$
between points inside the $|\phi|< \frac{\pi}{2}$ region,
the four possible sequences of alternating $G$ and $\sop  O_S$ in 
Eq (\ref{eq:possible_sequences}) just connect the value of $\phi$ after applying $G \sop O_* $
to the value of $\phi$ before the next application of $G \sop O_*$. Generalizing Figure \ref{fig:paths},
 one can see that the
shortest path uses either $ (G\sop O_S)^m$  or $(\sop O_SG)^m$
to connect points inside the  $|\phi|< \frac{\pi}{2}$ region. 
Therefore we do not need to consider  the sequences
$ G(\sop O_SG)^m$ or   $G(G\sop O_S)^m$. 

\begin{figure}%
\centering
\subfloat{
\includegraphics[width=6.2cm]{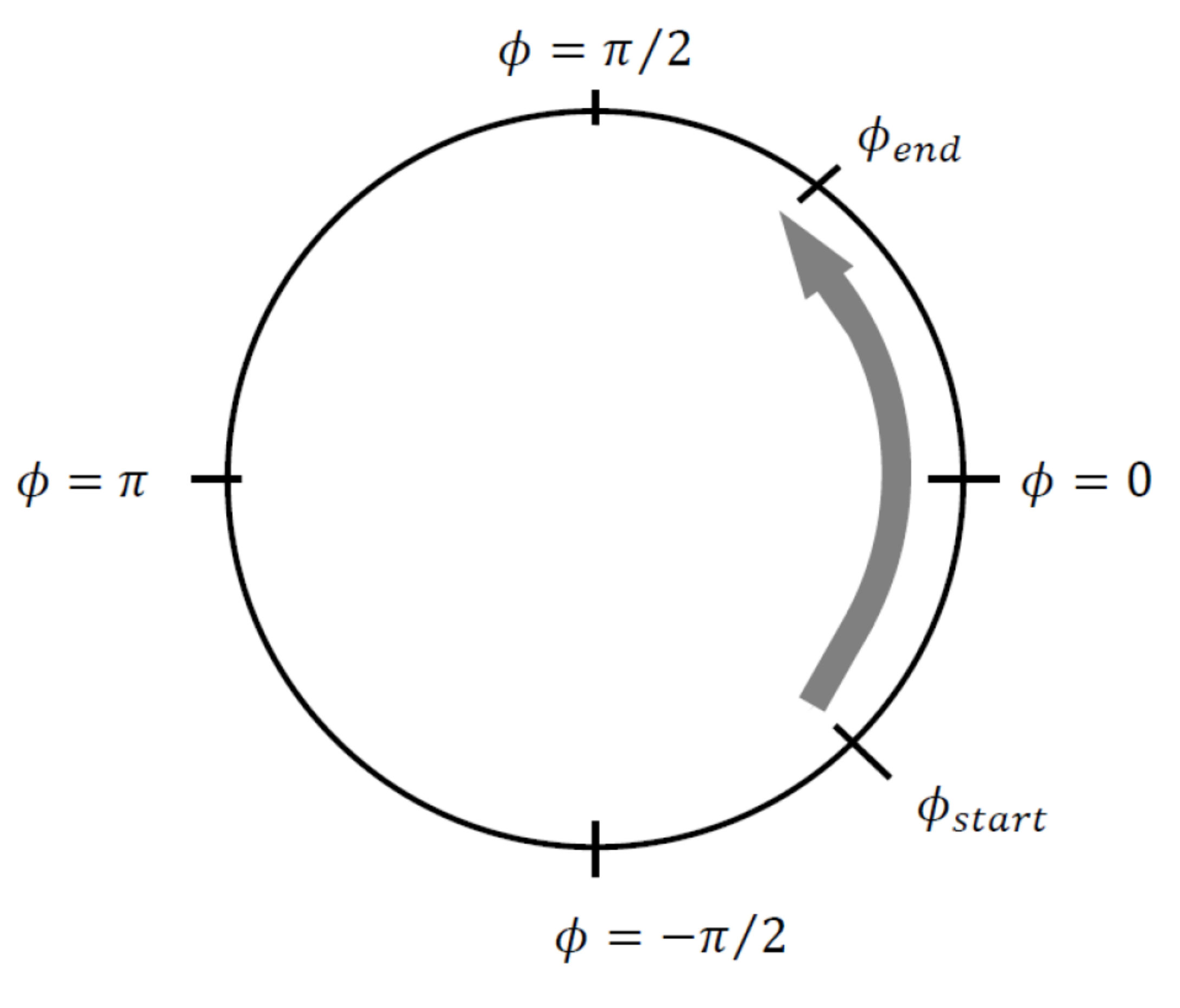}
\label{fig:shorpath}}
\qquad
\subfloat{\includegraphics[width=6.7cm]{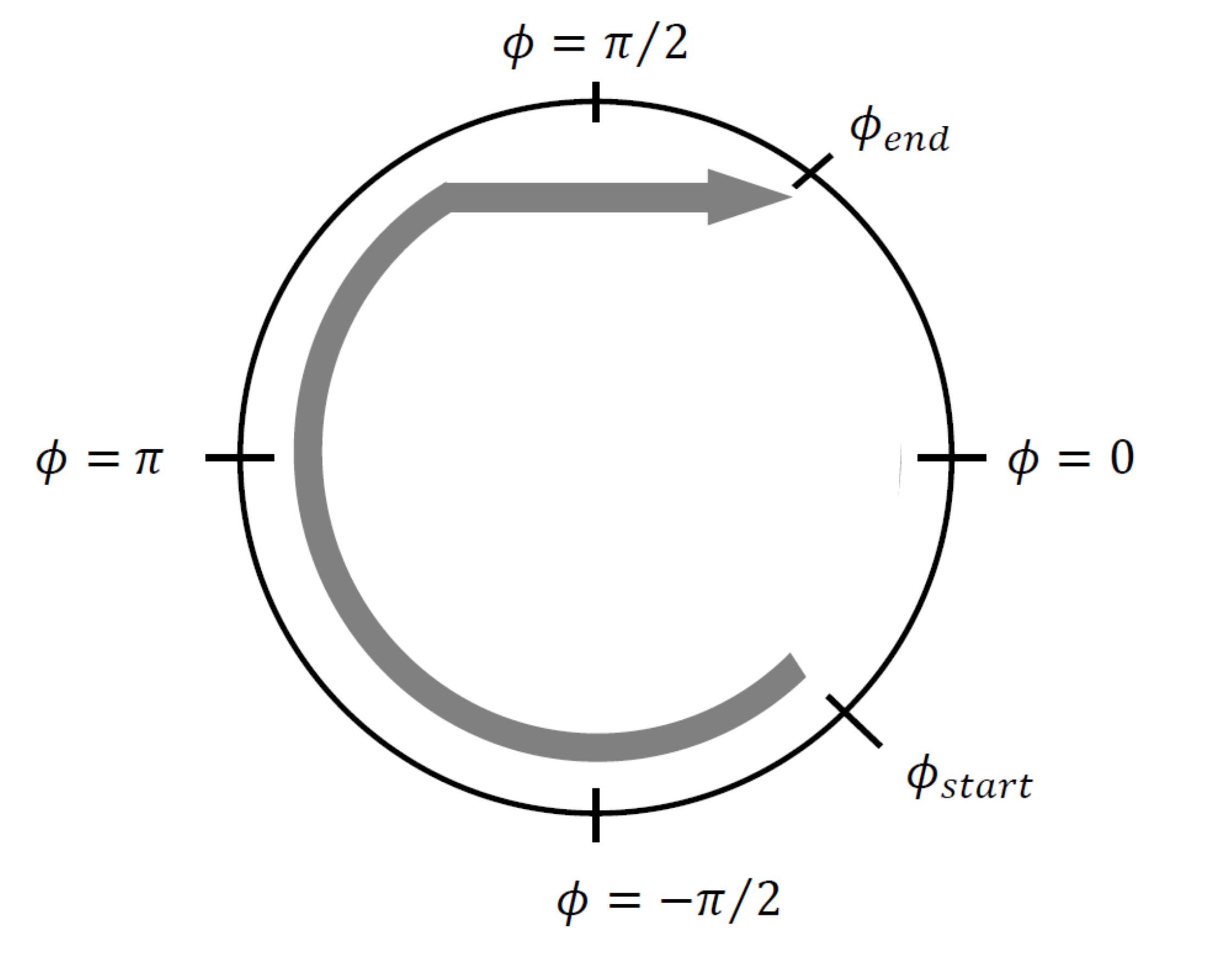}
\label{fig:longpath}}%
\caption{The path in the figure at left uses a sequence
$ (G\sop O_S)^m$ to move from $\phi_{start}$ to $\phi_{end}$,
whereas the path in figure at right uses a 
sequence $G(\sop O_SG)^m$. The path using $ (G\sop O_S)^m$
is shorter, signifying that
fewer uses of $\sop O_S$ are required to move from $\phi_{start}$ to $\phi_{end}$,
and thus this is the more efficient path.
}%
\label{fig:paths}%
\end{figure}



Next, we show that if one initially has $\phi>0$, it is never advantageous to 
again apply $G \sop O_*$ when $\phi<0.$
Since the algorithm must consist of applications of $ G \sop O_*$ separated by sequences 
of either $( \sop O_S G)^m$ or $(G \sop O_S )^m$, we can  enumerate and address the 
three possible cases that lead us to apply $\sop O_*$ at some
$\phi=\phi_{neg}<0$ after initially having $\phi\geq 0$. The three possible cases are laid out graphically in Figure \ref{fig:cases}.
In order to prove that none of the cases are optimal, we define the function 
\begin{align} \label{eq:progress_per_O_*}
 p_*(\phi) = 2(\theta_0 \sin(\phi+\phi_0) - k\phi)
 \end{align} as the change in progress 
function due to an application of
$\sop O_*$, dropping higher order terms. Note  for $\phi\geq0$, $\phi_{opt}$ optimizes Eq. (\ref{eq:progress_per_O_*}) as discussed after  Eq. (\ref{eq:cost_with_k}). We
proceed to treat the three cases.

\begin{enumerate}[\bfseries Sequence I.]
\item \label{case:increasephi}
We consider the following sequence of operations (see Figure \ref{fig:cases}): 
\begin{enumerate}[(i)]
\item  Start with $\phi_i>0$. Then apply $ G\sop O_* $ to get to $-\phi_i.$
\item  Apply $(G \sop  O_S )$ some number of times to increase $\phi$ to $\phi_{neg} > -\phi_i.$ 
\item Apply $G\sop O_*$ to get to $-\phi_{neg} < \phi_i.$
\end{enumerate}
The change in progress due only to $\sop O_*$ in this sequence is
\begin{align}
p_*(\phi_i)+p_*(\phi_{neg}) &= 2(\theta_0 \sin (\phi_i+\phi_0)-k\phi_i)\no
&+2(\theta_0 \sin (\phi_{neg}+\phi_0)-k\phi_{neg})  \no
 &\le 4[\theta_0 \sin (\frac{\phi_i+\phi_{neg}}{2}+\phi_0)
 -k\frac{\phi_{neg}+\phi_i}{2}]  \no
&= 2 p_*(\phi_i+\phi_{neg}) \no
& \le 2 p_*(\phi_{opt}),
\end{align}
Since $\phi_{neg} + \phi_i \ge 0$, the average progress due to the two applications of $\sop O_*$ 
is worse than if we had applied $\sop O_*$ at $\phi_{opt}$ both times. Thus this sequence
cannot be optimal.

\item \label{case:decreasephi}
We consider the following sequence of operations (see Figure \ref{fig:cases}): 
\begin{enumerate}[(i)]
\item  Start with $\phi_i > 0$. Then apply $G \sop O_* $ to get to $-\phi_i.$
\item  Apply $(\sop  O_S G )$ some number of times to decrease $\phi$ to $\phi_{neg}<-\phi_i.$
\item  Apply $G\sop O_*$ to get to $-\phi_{neg} > \phi_i$.
\end{enumerate}
Compare Sequence {\bf{II}} to the following Sequence ${\bf{2}}$:
\begin{enumerate}[(a)]
\item  Start with $\phi_i > 0$. Then apply $( G\sop  O_S  )$ 
some number of times to increase $\phi$ to $-\phi_{neg}>\phi_i.$
\end{enumerate}


The difference in progress between Sequence {\bf{II}} and Sequence {\bf{2}} is
\begin{align}
&(2\theta_0\sin(\phi_i+\phi_0)+2\theta_0\sin(\phi_{neg}+\phi_0))\no 
=&4\theta_0 \sin (\frac{\phi_i+\phi_{neg}}{2}+\phi_0)\cos(\frac{\phi_i-\phi_{neg}}{2})\nonumber\\
 <& 4 \theta_0 \sin{\phi_0}, 
\end{align}
since $-\frac{\pi}{4}<\frac{\phi_i+\phi_{neg}}{2}<0$ and
$0<\frac{\phi_i-\phi_{neg}}{2}<\frac{\pi}{2}$. 
Sequence {\bf{II}} and Sequence {\bf{2}} both use the same number of applications of $\sop O_S$ (in steps (ii) and (a) 
respectively). Therefore, the Sequence {\bf{II}} has an additional cost $2c_*$ while
it only has an added increase in progress of 
\begin{align}
4 \theta_0 \sin{\phi_0} =& 2p_*(0)  
\nonumber\\
\le& 2p_*(\phi_{opt}).
\end{align}
 Therefore Sequence {\bf{II}} does not attain the increase in progress per 
cost that one could attain by only applying $\sop O_*$ at $\phi_{opt}$.

\begin{figure}[b!]
\begin{minipage}[b]{1\linewidth}
\centering%
\fbox{\includegraphics[width=8cm]{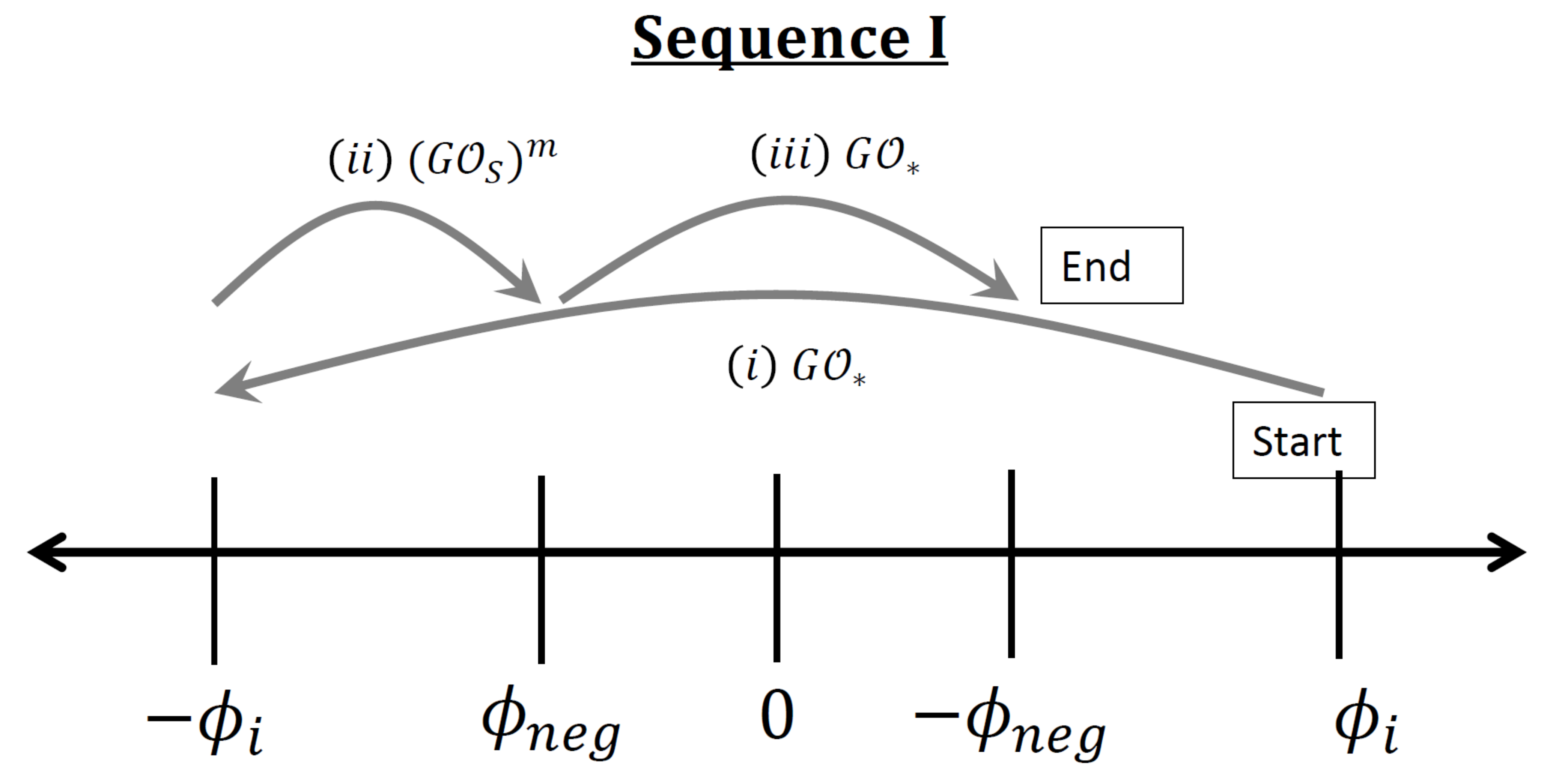}}
\end{minipage}%
\hfill%
\begin{minipage}[b]{1\linewidth}
\centering%
\fbox{\includegraphics[width=8cm]{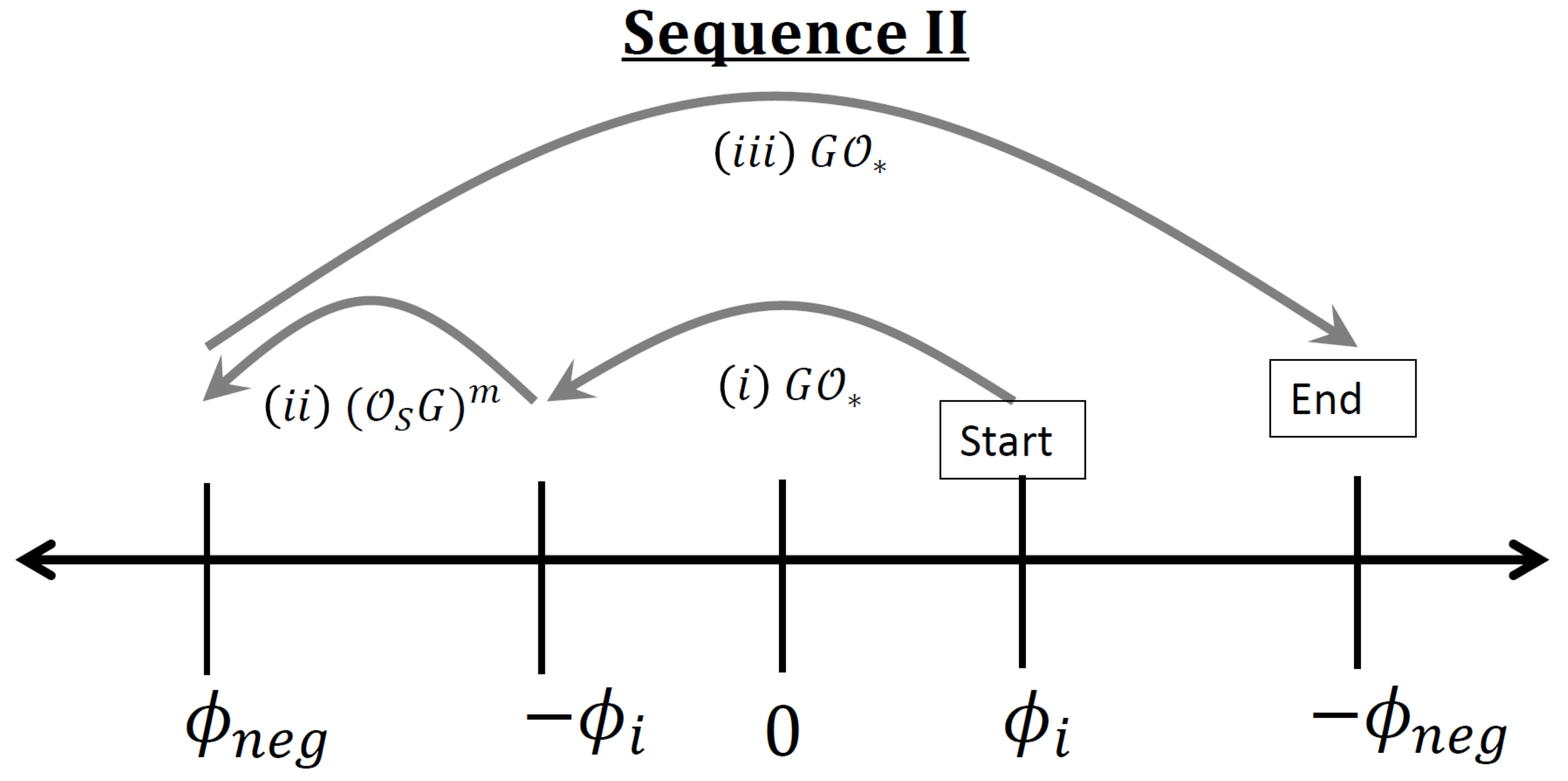}}
\end{minipage}%
\hfill%
\begin{minipage}[b]{1\linewidth}
\centering%
\fbox{\includegraphics[width=8cm]{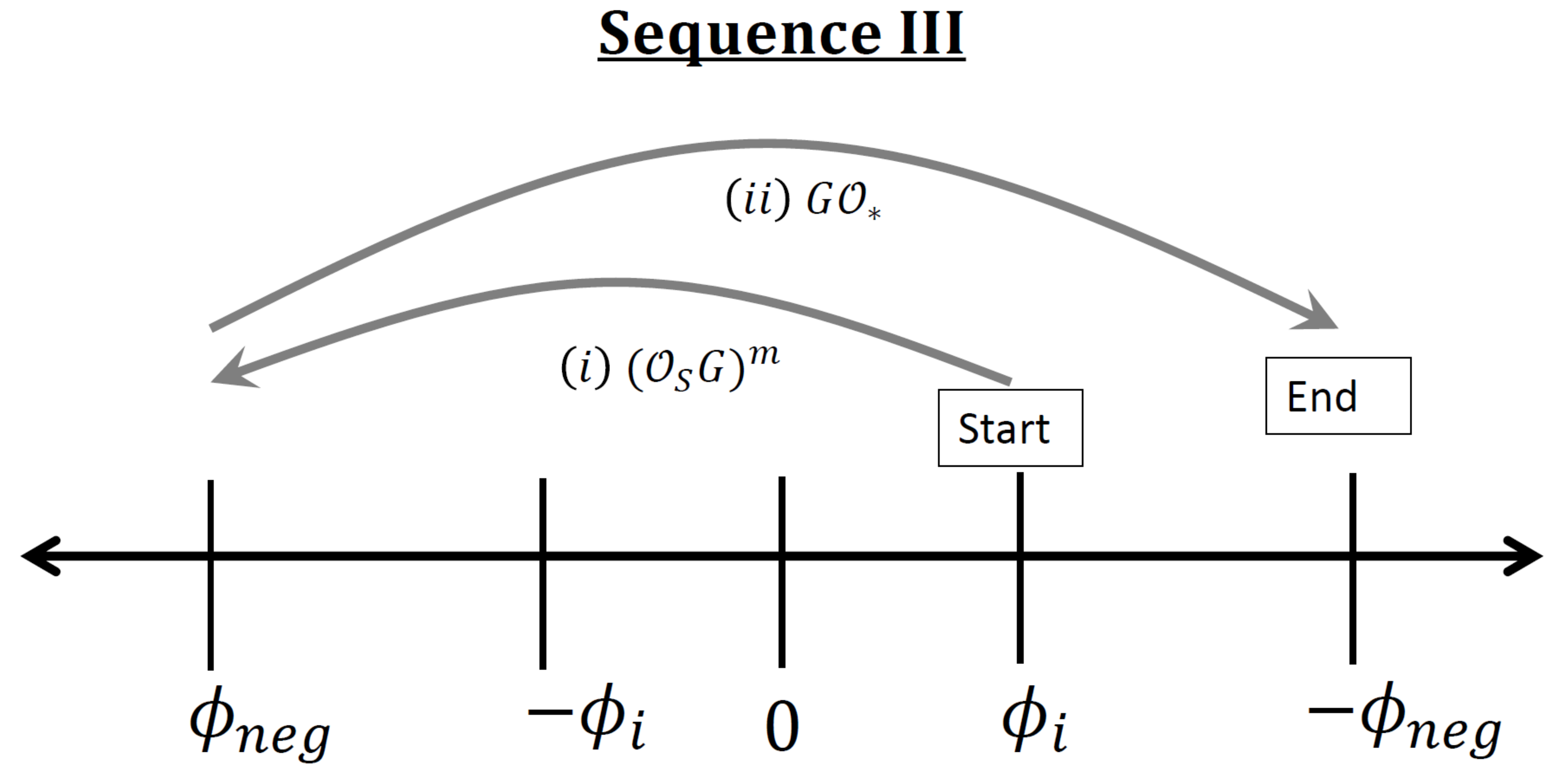}}
\end{minipage}
\caption{Possible paths that could lead to applying $G\sop O_*$ at a negative value 
of $\phi$, when initially, $\phi$ has positive value.}
\label{fig:cases}
\end{figure}

\item \label{case:straightdecreasephi}
We consider the following sequence of operations (see Figure \ref{fig:cases}): 
\begin{enumerate}[(i)]
\item  Start with $\phi_i \geq 0$,  then apply $(O_S  G)$ some number of times to decrease
 $\phi$ to $\phi_{neg}<0$.
\item Apply $G\sop O_*$ to get to $-\phi_{neg}$.
\end{enumerate}
Compare Sequence {\bf{III}} to the following Sequence {\bf{3}}:
\begin{enumerate}[(a)]
\item Start with $\phi_i \geq 0$, and then apply $(O_S  G)$ some number of times to decrease
 $\phi$ to $\phi_w$ such that $2\phi_0>\phi_w\geq0$.
\item  Apply $G \sop O_*$ to  get to $-\phi_w$.
\item  Apply $(GO_S)$ some number of times to increase $\phi$ to $-\phi_{neg}>0$.
\end{enumerate}
Note that we can always create a sequence with such a $\phi_w$ because $(O_S  G)$
changes $\phi$ by at most $2\phi_0$ each time.
The cost of Sequence {\bf{III}} is the same as the cost of Sequence {\bf{3}}. The difference in progress between 
Sequence {\bf{III}} and Sequence {\bf{3}} is
\begin{align}
&2\theta_0\sin(\phi_{neg}+\phi_0) - 2\theta_0\sin(\phi_w+\phi_0)\nonumber\\
\leq&4\theta_0 \cos\left(\frac{\phi_{neg}+\phi_w}{2}+\phi_0\right)\sin\left(\frac{\phi_{neg}-\phi_w}{2}\right)\no
 <& 0 
\end{align}
 since $|\frac{\phi_{neg}+\phi_w}{2}+\phi_0|<\frac{\pi}{2}$ and $\frac{\pi}{2}<\frac{\phi_{neg}-\phi_w}{2}<0$. Therefore  Sequence  \ref{case:straightdecreasephi} is not optimal either.
\end{enumerate}

Hence we conclude that applying $\sop O_*$ at negative $\phi$ never achieves
as much increase in progress per cost as applying $\sop O_*$ at $\phi_{opt}$,
and therefore we only need to consider applying $\sop O_*$ at positive $\phi$, at $\phi_{opt}$.
\end{proof}


\section{An Adversary Lower Bound}\label{sec:adversary}

In this section, we will show how to apply the adversary method to
the problem of cost complexity of STO. 



Suppose we are given access to an oracle $\sop O_*$, which implements the function
$f_*$, and an oracle $\sop O_S,$ which implements the function $f_S$. Then any algorithm
which solves STO
using these oracles, after $t$ steps, produces a state
\begin{align}
\ket{\psi_{f_*,f_S}^t}=U^t\sop O_{c_t}\cdots U^2
\sop O_{c_2} U^1\sop O_{c_1}\ket{\psi^0},
\end{align}
where $c_j\in\{*, S\},$ and $U^j$ are fixed unitaries independent of $f_*$ and $f_S.$


We create an adversary matrix $\Gamma$, a matrix
 whose rows and columns are indexed by pairs of functions $(f_*, f_S)\in D_{\textrm{STO}}$, where
 $D_{STO}$ is the set of valid inputs to STO. Furthermore, we have the condition that
 that $\Gamma[(f_*,f_S),(g_*,g_S)]=0$ if $\textrm{STO}(f_*,f_S)=\text{STO}(g_*,g_S).$
With this notation, we define the progress function:
\begin{align}
W^t=\sum_{(f_*,f_S),(g_*,g_S)\in D_{\textrm{STO}}\times D_{\textrm{STO}}}
\Gamma_{(f_*,f_S),(g_*,g_S)}v_{f_*,f_S}v_{g_*,g_S}^*\braket{\psi_{f_*,f_S}^t}{\psi_{g_*,g_S}^t}
\end{align}
for a vector $v$ indexed by the elements of $D_{\textrm{STO}}$, such that $\|v\|=1$ and
$v$ is an eigenvector of $\Gamma$ with eigenvalue $\pm\|\Gamma\|$, (where $\|\cdot\|$ signifies the 
$l$-2 norm for vectors or the induced $l$-2 norm for matrices).

Then following \cite{HLS07}\footnote{The proofs are identical, so we omit them.}, we have 
\begin{enumerate}
\item $W^0=\|\Gamma\|$.
\item $W^T\leq \left(2\sqrt{\epsilon(1-\epsilon)}+2\epsilon\right)\|\Gamma\|$, for any algorithm
with probability of error at most $\epsilon$.
\item $W^{t-1}-W^t\leq 2\max_i\|\Gamma\circ D_i^{c_t}\|$ where $D_i^{c_t}$ are $|D_{\textrm{STO}}|\times |D_{\textrm{STO}}|$
matrices satisfying
\begin{align}
&
D_i^*[(f_*,f_S),(g_*,g_S)]=
\begin{cases}
0 \text{ if } f_*(i)=g_*(i),\no
1 \text{ otherwise},
\end{cases}
&
D_i^S[(f_*,f_S),(g_*,g_S)]=
\begin{cases}
0 \text{ if } f_S(i)=f_S(i),\no
1 \text{ otherwise}.
\end{cases}
\end{align}
\end{enumerate}

Thus if $q_*$ queries are made to $\sop O_*$ and $q_S$ queries are made to $\sop O_S$,
we have
\begin{align}\label{eq:advbound}
\|\Gamma\|g(\epsilon)
\leq q_*\max_i\|\Gamma\circ D^*_i\|+q_S\max_i\|\Gamma\circ D^S_i\|
\end{align}
where
\begin{align}
g(\epsilon)=\frac{1-\left(2\sqrt{\epsilon(1-\epsilon)}+2\epsilon\right)}{2}.
\end{align}

We construct the following adversary matrix for STO:
$\Gamma[(f_*,f_S),(g_*,g_S)]=1$ if one of the following conditions holds:
\begin{itemize}
\item STO$(f_*,f_S)=1$, STO$(g_*,g_S)=0$, and $f_S(i)=g_S(i)$ except if $f_*(i^*)=1$, then $g_S(i^*)=0,$
\item STO$(g_*,g_S)=1$, STO$(f_*,f_S)=0$, and $g_S(i)=f_S(i)$ except if $g_*(i^*)=1$, then $f_S(i^*)=0.$
\end{itemize}
Otherwise, $\Gamma=0$. 

One can calculate (or it is easy to see by analogy to a standard Grover search over $N-M+1$
items)
that 
\begin{align}
\|\Gamma\|&=\sqrt{N-M+1},\no
\max_i\|\Gamma\circ D_i^{c_t}\|&=1,\no
\max_i\|\Gamma\circ D^S_i\|&=1.
\end{align}
Plugging into Eq. (\ref{eq:advbound}) we have 
\begin{align}
g(\epsilon)\sqrt{N-M+1}
\leq q_*+q_S,
\end{align}
so for $N>M/2$, we have
\begin{align}\label{eq:adv_bound1}
QCC(\textrm{STO})=\Omega(c_S\sqrt{N}).
\end{align}

We also consider a second adversary matrix for STO. Let 
$\Gamma[(f_*,f_S),(g_*,g_S)]=1$ if one of the following conditions holds:
\begin{itemize}
\item STO$(f_*,f_S)=1$, STO$(g_*,g_S)=0$, and $f_S(i)=g_S(i)$,
\item STO$(g_*,g_S)=1$, STO$(f_*,f_S)=0$, and $g_S(i)=f_S(i)$.
\end{itemize}
Otherwise, $\Gamma=0$.

In this case, the adversary matrix
only pairs instances such that  $\sop O_S$ is the same in both pairs. Thus it is as if the set $S$
is known ahead of time.
In this case, one can calculate (or it is easy to see by analogy to a standard Grover search over $M$
items), that
\begin{align}
\|\Gamma\|&=\sqrt{M}\no
\max_i\|\Gamma\circ D_i^{j_t}\|&=1\no
\max_i\|\Gamma\circ D^S_i\|&=0.
\end{align}
Plugging into Eq. (\ref{eq:advbound}), we have
\begin{align}
g(\epsilon)\sqrt{M}\leq q_*,
\end{align}
so 
\begin{align}\label{eq:adv_bound2}
QCC(\textrm{STO})=\Omega(c_*\sqrt{M})
\end{align}
Combining Eq. (\ref{eq:adv_bound1}) and Eq. (\ref{eq:adv_bound2}), we obtain a bound that matches Eq. (\ref{eq:cost12}):
\begin{align}
QCC(\textrm{STO})=\Omega\left(\max\{c_*\sqrt{M},c_S\sqrt{N}\}\right).
\end{align}

\end{document}